%% file: main.tex
\documentclass[preprint,12pt,nopreprintline]{elsarticle}
\usepackage{etoolbox}
\usepackage{amsmath,amssymb,booktabs,array,longtable,tabularx}
\usepackage[hidelinks]{hyperref}
\usepackage{xcolor,needspace}
\AtBeginDocument{\color{black}}
\journal{European Journal of Control}

\begin{document}
\begin{frontmatter}
\input{title.tex}
\input{author.tex}
\begin{abstract}
This paper numerically investigates nonlinear model construction from
supplied structural information within predictive cost adaptive control.
The study varies the retained nonlinear functions, coefficient sharing,
and forward reconstruction under a common identification and control
architecture. Recursive least squares updates unknown coefficients from
direct or integrated relations that are linear in the parameters.
The reconstructed nonlinear forward predictor need not be parameter-linear.
A local affine approximation, constructed from the predictor value and
Jacobian, is frozen over the prediction horizon for constrained quadratic
optimization. Scalar full-state studies compare exact parameterizations,
Taylor-informed approximations, broader dictionaries, and incorrectly
restricted models. The comparisons include parameter variation, structural
changes, and selected noisy cases. Common-data tests assess nonlinear
prediction, whereas common-query forecasts distinguish nonlinear-model
error from affine-freezing error. Under the tested conditions, correct
structure improves prediction relative to affine models without increasing
parameter count. Integrated physical identification reduces prediction
error relative to direct sampled-map fitting after initialization.
Additional coefficients and improved nonlinear prediction do not ensure
lower tracking error. Forgetting improves readaptation without ensuring
parameter recovery, and coefficient activation does not establish structure
discovery. Affine freezing can reverse the nonlinear prediction advantage
over an affine baseline. The results support evaluating representation,
identification, reconstruction, and horizon prediction separately before
increasing model complexity.
\end{abstract}

\begin{highlights}
\item Supplied nonlinear structure is compared within a common PCAC architecture.
\item Prediction, tracking, and parameter recovery are assessed separately.
\item Physical identification is distinguished from nonlinear forward reconstruction.
\item Affine freezing can reverse a nonlinear model's prediction advantage.
\end{highlights}

\end{frontmatter}

\input{introduction.tex}
\input{control_problem.tex}
\input{nonlinear_models.tex}
\input{identification.tex}
\input{predictive_control.tex}
\input{numerical_study.tex}

\input{discussion_conclusions.tex}

\appendix
\input{numerical_protocol.tex}

\section*{Acknowledgments}
This research was partially supported by the MEXT SPReAD AI for Science program 2026 (Grant No. 26276994).

\input{main.bbl}
\end{document}

%% file: title.tex
\title{A Numerical Investigation of \\ Indirect Adaptive Predictive Control with \\ Structure-Informed Nonlinear Regressors}

%% file: author.tex
\author{Tam W. Nguyen}
\ead{nguyen.tamwilly.3e@kyoto-u.ac.jp}
\affiliation{organization={Department of Electrical Engineering, Graduate School of Engineering, Kyoto University},
            addressline={A1-411, Kyoto Katsura, Nishikyo-ku},
            city={Kyoto},
            postcode={615-8510},
            state={Kyoto},
            country={Japan}}

%% file: introduction.tex
\section{Introduction}

Indirect adaptive predictive control identifies a model online and uses it
for prediction and finite-horizon optimization. Predictive cost adaptive
control (PCAC) combines recursive least-squares (RLS) identification with
constrained receding-horizon optimization~\cite{nguyen2021predictive}.
The PCAC predictive formulation is related to generalized predictive
control~\cite{clarke1987generalizedpart1,clarke1987generalizedpart2}.
PCAC has been investigated for flutter suppression~\cite{richards2025predictive}
and active flow control~\cite{VanderSchaafEtAl2026ActiveFlow}.
Experimental implementations include quadcopter flight~\cite{richards2025experimental}
and acoustic noise suppression~\cite{KamaldarEtAl2024NoiseVibration}.
PCAC updates the model during operation to accommodate uncertainty and
changes in the plant. For nonlinear plants, the design question concerns
both the estimator and the structure supplied to the identification model.
A complete numerical model may be unavailable even when some functional
relations are known. This paper investigates
how retaining such information affects prediction and closed-loop
control.

Nonlinear identification includes models combining linear dynamics with
static nonlinear functions~\cite{BILLINGS198215}, and noise-aware
identification has also been studied for block-oriented nonlinear
models~\cite{LiEtAl2025MIMOHammersteinNoise,LiEtAl2023HammersteinOutputNoise}.
Finite nonlinear regressors permit models that are nonlinear in measured
variables but linear in unknown coefficients~\cite{SchoukensLjung2019NonlinearID}.
The direct-history model in Sections~3.1 and~3.3 is a parameter-linear nonlinear
autoregressive model with exogenous input (NARX), rather than a distinct
identification class, and need not impose a Hammerstein cascade.
Alhazmi et al.~\cite{Alhazmi2026NonlinearPCAC} identify pseudo-linear input--output
models with finite polynomial, Fourier, and cubic-spline observables using RLS.
Their iterative model predictive control (MPC) step updates the horizon
prediction model using predicted trajectories. The present study compares model
construction under a common
Jacobian-frozen affine control step, without inner updates of the horizon model. The dictionary need not discard
available plant information. Known contributions can remain fixed, and known
nonlinear functions can be retained explicitly. Sharing unknowns can
impose coefficient relations without estimating a separate
coefficient for every term. When the supplied functions and relations
give an exact parameterization, identification concerns only the
independent unknown parameters. Exact representability establishes that
a suitable parameter vector exists, but does not establish that the
available data identify that vector accurately.

For RLS, linearity in the unknown parameters is required of the
identification relation, not of the reconstructed nonlinear forward
predictor. The identification relation and forward predictor need not
have the same form. Physical equations may supply an inverse or
interval-integrated relation from which the parameters are estimated.
Physical-identification methods exploit
relations that are nonlinear in measured motion but linear in inertia
entries~\cite{ManchesterPeck2017InertiaEstimation}.
Interval integration can avoid estimating state
derivatives~\cite{ParikhEtAl2019IntegralCL}.
The scalar physical example uses an integrated
parameter-linear relation and reconstructs the nonlinear forward map
by propagating the identified differential equation. An exact physical
parameterization does not make the quadrature or numerical forward
propagation exact. Taylor expansion is unnecessary
when a computable parameter-linear relation and a usable nonlinear
forward model are available.

Partial structural information can restrict an approximate predictor
when the exact functional form is unavailable.
Taylor/Maclaurin structure provides one local construction under
sufficient smoothness~\cite{Nguyen2026TaylorJacobian}. Known sparsity,
symmetries, and couplings restrict the retained monomials and
coefficient relations. The fitted coefficients need not equal the
derivatives of the true map. Exact structural eliminations, finite-order
truncation, and additional deliberate reduction are distinct.
With little reliable functional information, a broader candidate
dictionary remains an alternative, although its variables, memory, and
functions must still be specified. A generic dictionary does not impose the
particular structural restrictions being investigated.
The amount of supplied information and the correctness of its
restrictions are separate considerations. An approximation can retain
valid structure, whereas an incorrect restriction can exclude dynamics
needed for prediction.

Related linear parameter-varying (LPV) identification methods address unknown
scheduling dependence. T{\'o}th et al.~\cite{TothEtAl2011LPVStructureLearning} use least-squares
support vector machines (LS-SVM) to estimate scheduling-dependent coefficient
functions in LPV-ARX models through kernel representations.
Rizvi et al.~\cite{RizviEtAl2018KernelizedLPV} first reconstruct a state sequence
from input, output, and scheduling data using kernel canonical correlation
analysis, and then identify scheduling-dependent state-space matrices with
LS-SVM. Recursive LPV-ARX identification with informative-data selection is also
available~\cite{CorriniEtAl2025OnlineLPVARX}.
Here, RLS updates a fixed finite vector of independent coefficients in explicitly
supplied regressors, including specified coefficient-sharing relations. Kernel
choices can also incorporate prior information~\cite{TothEtAl2011LPVStructureLearning}, and
finite-feature kernel models can overlap with explicit regressions. The
distinction is therefore the chosen parameterization and estimation procedure,
not an exclusive ability to use prior information or update online. The functions
and coefficient-sharing relations here are not selected online.

Other formulations of adaptive MPC combine recursive
model updates with robust constraint treatment~\cite{LorenzenEtAl2019RecursiveModelMPC},
RLS estimation with tube-based prediction~\cite{ZhangShi2020AdaptiveMPC},
or uncertainty compensation with tightened constraints~\cite{TaoEtAl2026UncertaintyCompensation}.
Buerger and Cannon combine successive linearization with ellipsoidal
tubes in a robust adaptive nonlinear MPC
formulation~\cite{BuergerCannon2026RobustAdaptiveNMPC}.
Explicit excitation constraints have also been used to support parameter
identification~\cite{LuCannon2023ExcitationMPC}. The present study uses the
same PCAC predictive-control step for all model constructions and does
not propagate uncertainty sets or enforce excitation constraints.
The controller evaluates and differentiates the reconstructed nonlinear
map at the current operating point, using its value and Jacobian to
define a local affine predictor. The affine offset is retained, so the
operating point need not be an equilibrium. The affine coefficients are held fixed over the
prediction horizon and recomputed at the next control instant.
A quadratic tracking objective, hard input and input-increment bounds,
and optional soft output bounds yield a constrained quadratic program
(QP). Unlike nonlinear model predictive
control~\cite{rawlings2020model}, this QP does not optimize a recursive
nonlinear rollout. Agreement between the nonlinear and affine maps at
the current point does not establish agreement over the horizon.
The approximation introduced by affine freezing is therefore assessed
separately from nonlinear-model error. Predicted output bounds do not
certify bounds on the actual plant output.

This paper contributes a structured numerical investigation of nonlinear
regressor construction under the common PCAC architecture.
It does not propose a new nonlinear-regression class or local-linearization
principle. The comparisons assess how supplied structure affects nonlinear
prediction and closed-loop tracking, while evaluating parameter recovery
separately.
Reference \cite{Nguyen2026TaylorJacobian} develops Taylor-informed polynomial models with
Jacobian-frozen affine predictors and compares polynomial degrees.
Here, the Taylor-informed model is one choice alongside exact-function
and physical-reconstruction models.
We vary which functions are retained, which coefficients are estimated
independently or shared, and how the nonlinear forward predictor is
reconstructed from the estimates. The comparisons distinguish exact
parameterizations, structure-informed
approximations, broader candidate dictionaries, and incorrectly
restricted models, with affine PCAC and known-model references.
Stationary studies separate the value of supplied structure from the
number of estimated coefficients and examine approximation order over
prescribed operating regions. Further studies consider parameter
drift, abrupt changes, and nonlinear plant changes within and outside
the retained model class. A physical example compares integrated
identification and nonlinear reconstruction with direct sampled-map
fitting.

Common-data identification tests use common fitting records and
independent evaluation records. Closed-loop tests assess adaptive
controllers whose actions may generate different identification data.
Forecast diagnostics use common initial conditions and input sequences
to separate nonlinear-model error, affine-freezing error, and total
prediction error over several horizons. Tracking, control activity, and
readaptation are evaluated separately from parameter recovery.
Selected comparisons include measurement noise and process
disturbances, and numerical conditioning is monitored.
RLS is compared without forgetting, with fixed forgetting, and with
residual-triggered variable-rate
forgetting~\cite{Bruce2020VRFRLS}. The dictionaries remain
fixed. Forgetting reweights past data but does not generate excitation
or supply missing functions. Changes in fitted nonlinear coefficients
are therefore not treated as evidence of online structure discovery.

The numerical studies use scalar sampled recurrences and a scalar
differential equation with full-state measurements. This design
isolates model construction and adaptation from state-estimation
errors. The formulation also admits finite measured histories, but
a chosen history length does not establish an exact representation of
an arbitrary partially measured plant. State estimation and sensor
realization are outside the scope. The known-model references retain
the same frozen-affine QP and are diagnostic comparisons, not
performance bounds. The investigation does not establish general closed-loop stability or
parameter convergence. Section~5.3 gives a feasible-point construction
for the implemented QP under stated conditions.
The aim is to provide conditional design guidance by distinguishing
limitations of representation, identification, reconstruction, and
horizon prediction before introducing additional model complexity.

%% file: control_problem.tex
\section{Control Problem and Available Information}
\label{sec:c4-problem}

\subsection{Plant and available information}
Consider a causal nonlinear plant with input $u_k\in\mathbb R^m$ and
measured output $y_k\in\mathbb R^p$. Let $T_s>0$ be the sampling period,
and let $t_k=kT_s$ denote the sampling instants. In the absence of disturbances and
measurement noise, denote the plant's input--output relation by
\begin{equation}
    y_{k+1}=\mathcal P_k(y_{-\infty:k},u_{-\infty:k}),
    \label{eq:c4-plant}
\end{equation}
where $y_{-\infty:k}$ and $u_{-\infty:k}$ are the output and input
histories through index $k$, respectively. The unknown operator
$\mathcal P_k$ may depend on $k$, which permits but does not require
time-varying plant dynamics.
Equation~\eqref{eq:c4-plant} does not prescribe a finite memory or a
state-space realization.

A complete numerical plant model is not supplied, although structural
information about the plant may be supplied. A plant equilibrium is
not assumed known.

\subsection{Prediction model and control objective}

The controller uses a finite-dimensional state-space model to predict
future outputs. Let $\xi\in\mathbb R^n$ denote the state variable of
this model. At each control instant $t_k$, the initial state $\xi_k$ of
the prediction trajectory is constructed from data available through $t_k$.
The two constructions considered here use either a finite measured
input--output history or measurements of the full physical state.

For the history construction, choose an integer $\ell\geq1$
and define
\begin{equation}
    \xi_k=s_k:=
    \operatorname{col}
    (y_k,\ldots,y_{k-\ell+1},
     u_{k-1},\ldots,u_{k-\ell+1}),
    \label{eq:c4-history}
\end{equation}
where $n=p\ell+m(\ell-1)$. The vector $s_k$ stores measured outputs
and past inputs, with an empty input block when $\ell=1$.
The required history is supplied or collected before predictive control
begins. This construction does not require a physical state-space model.
Selecting $\ell$ specifies the memory of the prediction model but does
not establish that the plant admits an exact representation using this
history.

\Needspace{11\baselineskip}
For the full-state construction only, assume that the measured output
contains all components of a physical plant state.
Set $n=p$, let $x_k\in\mathbb R^n$ denote that state, and write
\[
    y_k=x_k+v_k,
\]
where $v_k\in\mathbb R^n$ is measurement noise.
The prediction trajectory starts from the measured state,
\[
    \xi_k=y_k,
\]
so $\xi_k=x_k$ in the noise-free case. No state observer is introduced.

For either construction, define the fixed output-selection matrix
$C\in\mathbb R^{p\times n}$ to select the output coordinates from the
prediction state. For the history construction, $C=[I_p\;0]$, whereas
for the full-state construction, $C=I_n$.
At initialization, the measured output satisfies
\begin{equation}
    y_k=C\xi_k.
    \label{eq:c4-output}
\end{equation}
The horizon starts from $\xi_{0|k}=\xi_k$, and the predicted outputs
are $y_{i|k}=C\xi_{i|k}$.

During prediction, the history model predicts a new output and shifts
the stored history through known shift equations.
The full-state model propagates the modeled physical-state dynamics.
The nonlinear state recursion for each construction is defined in the
following section.

Let $r_k\in\mathbb R^p$ be the output reference. The control objective is
to track $r_k$ while respecting prescribed input and input-increment
bounds.

At each control instant $t_k$, the input $u_k$ is fixed for the current
sampling interval. Using data available through $t_k$, including the
measured output $y_k$, the controller computes $u_{k+1}$ for application
at $t_{k+1}$. The initial input $u_0$ is prescribed. The completed
transition ending at $t_k$ was generated by $u_{k-1}$.

%% file: nonlinear_models.tex
\section{Structure-Informed Nonlinear Models}
\label{sec:c4-models}

We specify which functional relations are retained, which numerical
quantities are unknown, and which coefficient relations are imposed.
These specifications determine the regressor and remain fixed during
online coefficient identification. The amount of supplied information
and the correctness of its restrictions are distinct properties.

\input{modeling_choices.tex}

\subsection{Known functions and unknown coefficients}
Consider the scalar, noise-free relation
\begin{equation}
    y_{k+1}=d+a y_k+b\sin u_k,
    \label{eq:c4-sine-plant}
\end{equation}
where $d$, $a$, and $b$ are unknown. Suppose the functional form in
\eqref{eq:c4-sine-plant} is supplied. Define
\begin{equation}
    \phi(y,u)=\begin{bmatrix}1&y&\sin u\end{bmatrix}^{\!\top},
    \qquad
    \theta=\begin{bmatrix}d&a&b\end{bmatrix}^{\!\top}.
    \label{eq:c4-sine-regressor}
\end{equation}
Then $y_{k+1}=\phi(y_k,u_k)^\top\theta$. The nonlinear function is
retained, and only its unknown coefficients are identified. The origin need
not be an equilibrium because $d$ is unknown.

More generally, define the direct output model
\begin{equation}
    g(s,u;\theta)=g_K(s,u)+\Psi(s,u)\theta,
    \label{eq:c4-direct-model}
\end{equation}
where $g_K$ is a supplied known contribution,
$\Psi(s,u)\in\mathbb R^{p\times n_\theta}$ is the supplied regressor
matrix, and $\theta\in\mathbb R^{n_\theta}$ contains the independent
unknown coefficients. Set $g_K=0$ when no contribution is retained
exactly. The retained contribution belongs to the sampled output relation;
a known continuous-time term is not a known sampled contribution without
a corresponding derivation. The regressor may depend nonlinearly on $s$
and $u$, but it is independent of $\theta$.

Known coefficient relations can be retained within the regressor.
For example, for
\begin{equation}
    g(y,u;\theta)=a y+b(u+c y u),
    \qquad
    \theta=\begin{bmatrix}a&b\end{bmatrix}^{\!\top},
    \label{eq:c4-sharing}
\end{equation}
a known $c$ gives $\Psi(y,u)=[y\;\;u+c y u]$.
Estimating separate coefficients of $y$, $u$, and $yu$ relaxes the
relation between the last two coefficients. Both models contain the
interaction $yu$, but only the first retains the specified relation.

The completed measured transition supplies
\begin{equation}
    b_k=y_k-g_K(s_{k-1},u_{k-1}),
    \qquad
    \Phi_k=\Psi(s_{k-1},u_{k-1}).
    \label{eq:c4-direct-data}
\end{equation}
When a constant parameter vector represents the plant exactly and the
data are ideal, $b_k=\Phi_k\theta$. Otherwise, this equality has a
regression residual. Selecting a finite history and a regressor does not
by itself establish exact representability.

\subsection{Partial structure and local approximation}
Suppose less information is supplied about
\eqref{eq:c4-sine-plant}. An affine dependence on $y$ and an odd input
contribution are retained, but the sinusoidal function is not supplied. A local
polynomial model about $u=0$ is
\begin{equation}
    g_D(y,u)=d+a y+
    \sum_{\substack{1\leq j\leq D\\ j\ \mathrm{odd}}}b_j u^j.
    \label{eq:c4-odd-model}
\end{equation}
The retained powers follow from the supplied oddness, and their coefficients
are unknown. No relation between different $b_j$ is imposed. A finite
polynomial does not represent $\sin u$ exactly on a nontrivial interval.
Exact representation is not required to define the approximate predictor.

For the general Taylor construction, let $g$ be the nominal finite-history
map being approximated and write $g=g_K+h$. In a direct state-based
construction, $x$ replaces $s$. Define
\begin{equation}
    z=\operatorname{col}(s,u),\qquad
    \zeta=z-z_\circ,\qquad n_z=\dim z,
    \label{eq:c4-expansion-coordinates}
\end{equation}
where $z_\circ$ is a fixed expansion point. Assume $h$ is of class
$C^{D+1}$ on a neighborhood of $z_\circ$. For a multi-index
$\alpha\in\mathbb N_0^{n_z}$, define
\begin{equation}
    |\alpha|=\sum_{j=1}^{n_z}\alpha_j,
    \qquad \alpha!=\prod_{j=1}^{n_z}\alpha_j!,
    \qquad \zeta^\alpha=\prod_{j=1}^{n_z}\zeta_j^{\alpha_j},
    \qquad \mathcal A_D=\{\alpha:|\alpha|\leq D\}.
    \label{eq:c4-indices}
\end{equation}
The local expansion is
\begin{equation}
    h(z)=\sum_{\alpha\in\mathcal A_D}h_\alpha\zeta^\alpha+r_D(z),
    \qquad h_\alpha=\frac{1}{\alpha!}\partial^\alpha h(z_\circ).
    \label{eq:c4-taylor}
\end{equation}
This expansion applies to the specified smooth map. It does not establish
that an arbitrary partially measured plant admits such a finite-history
map.

Let $\mathcal A\subseteq\mathcal A_D$ be the index set retained under the
supplied structural restrictions. Define
\begin{equation}
    \phi_{\mathcal A}(z)=
    \begin{bmatrix}\zeta^{\alpha_1}&\cdots&\zeta^{\alpha_q}\end{bmatrix}^{\!\top},
    \qquad q=|\mathcal A|.
    \label{eq:c4-features}
\end{equation}
Each monomial occurs once in a fixed ordering. Let
$\Theta_{\mathcal A}\in\mathbb R^{p\times q}$ collect the true retained
Taylor coefficients. Then
\begin{align}
    g(z)&=g_K(z)+\Theta_{\mathcal A}\phi_{\mathcal A}(z)+r_{\mathcal A}(z),
    \label{eq:c4-retained-model}\\
    r_{\mathcal A}(z)&=r_D(z)+
    \sum_{\alpha\in\mathcal A_D\setminus\mathcal A}
    h_\alpha\zeta^\alpha.
    \label{eq:c4-pruning-error}
\end{align}
An exact structural restriction may make an omitted coefficient zero.
A deliberate omission need not do so. The second term in
\eqref{eq:c4-pruning-error} may therefore contain powers of degree no
greater than $D$. Truncation and additional structural reduction are
separate approximations~\cite{Nguyen2026TaylorJacobian}.

Let $S_\theta\in\mathbb R^{pq\times n_\theta}$ encode the imposed linear
coefficient relations and define
\begin{equation}
    \operatorname{vec}(\Theta(\theta))=S_\theta\theta,
    \qquad
    \Psi(z)=\bigl(\phi_{\mathcal A}(z)^\top\otimes I_p\bigr)S_\theta.
    \label{eq:c4-coefficient-relations}
\end{equation}
The operator $\operatorname{vec}$ stacks columns. The columns of
$S_\theta$ are linearly independent. Fixed nonzero coefficients are
included in $g_K$. Independent coefficients correspond to $S_\theta=I_{pq}$.
Equation~\eqref{eq:c4-direct-model} then gives the identification model.
An incorrect coefficient restriction introduces additional model
discrepancy.

The expansion point is not required to be an equilibrium. The constant
monomial is retained unless its value is supplied and accounted for in
$g_K$. The fitted coefficients need not equal the derivatives in
\eqref{eq:c4-taylor}; they also depend on truncation, restrictions, and
the identification data. A generic polynomial model relaxes the specified
restrictions within a declared monomial set. Here, generic does not mean
random.

\subsection{Forward prediction from measured histories}
The output model predicts one new output. The remaining history is
shifted without identification. For $\ell\geq2$, define
\begin{equation}
    F(s_k,u_k;\theta)=\operatorname{col}\bigl(
    g(s_k,u_k;\theta),
    y_k,\ldots,y_{k-\ell+2},
    u_k,\ldots,u_{k-\ell+2}\bigr).
    \label{eq:c4-history-map}
\end{equation}
For $\ell=1$, set $F(s,u;\theta)=g(s,u;\theta)$.
Equation~\eqref{eq:c4-history-map} defines the same shift operation for
measured histories and predicted histories. Thus, the nonlinear predictor
state recursion is
\begin{equation}
    \xi_{k+1}=F(\xi_k,u_k;\theta),\qquad
    y_{k+1}=C\xi_{k+1}.
    \label{eq:c4-forward-model}
\end{equation}
Only the new-output block
contains identified coefficients. The shift rows are known exactly.

\subsection{Forward prediction from a retained physical model}
When the physical state is measured, a known differential equation may
instead define $F$ through propagation over one sampling interval.
The equation used for identification need not be this forward equation.

Consider the scalar nonlinear model
\begin{equation}
    J\dot x=u-dx^3,\qquad J>0,\quad d\geq0,
    \label{eq:c4-inverse-example}
\end{equation}
where $J$ and $d$ are unknown and the state is measured as $y_k=x_k+v_k$.
Suppose the parameters are constant on $[t_{k-1},t_k]$ and the applied
input is $u_{k-1}$. Integration gives
\begin{equation}
    T_su_{k-1}
    =\begin{bmatrix}
       x_k-x_{k-1}&\displaystyle\int_{t_{k-1}}^{t_k}x(t)^3\,dt
     \end{bmatrix}
     \begin{bmatrix}J\\d\end{bmatrix}.
    \label{eq:c4-integrated-exact}
\end{equation}
Define $\theta=[J\;d]^\top$. The trapezoidal rule approximates the interval
integral from endpoint measurements. The implemented regression quantities are
\begin{equation}
    b_k=T_su_{k-1},\qquad
    \Phi_k=\begin{bmatrix}
      y_k-y_{k-1}&\dfrac{T_s}{2}(y_{k-1}^3+y_k^3)
    \end{bmatrix}.
    \label{eq:c4-integrated-data}
\end{equation}
This regression is linear in $\theta$. Measurement and quadrature errors
contribute to its residual.

For a specified $\theta$, define $X$ by
\begin{equation}
    \dot X(\tau)=\frac{u-dX(\tau)^3}{J},
    \qquad X(0)=x,\qquad
    F(x,u;\theta)=X(T_s),
    \label{eq:c4-physical-forward}
\end{equation}
where $u$ is constant for $0\leq\tau\leq T_s$.
This forward map need not be linear in $\theta$. Its numerical evaluation
may require integration. An exact continuous-time parameterization does
not make an approximate quadrature or numerical flow exact.

Related physical-identification methods use relations that are nonlinear
in measured motion but linear in inertia
entries~\cite{ManchesterPeck2017InertiaEstimation}. Each such construction must supply
its measurable regression and its forward map. No usable parameter-linear
regression is presumed for every nonlinear system. State-based prediction
uses the measured state in specified coordinates. For a geometric state,
valid local coordinates and the linear selection of tracked coordinates
must be specified before applying the formulation below.

%% file: modeling_choices.tex
Table~\ref{tab:c4-modeling-choices} summarizes the model-construction choices.

\begin{table}[!ht]
\centering
\caption{Model construction from the supplied structural information.}
\label{tab:c4-modeling-choices}
\small
\renewcommand{\arraystretch}{1.15}
\begin{tabularx}{\linewidth}{@{}%
    >{\raggedright\arraybackslash}p{0.25\linewidth}
    >{\raggedright\arraybackslash}p{0.39\linewidth}
    >{\raggedright\arraybackslash}X@{}}
\toprule
\textbf{Available information} &
\textbf{Identification choice} &
\textbf{Nonlinear predictor} \\
\midrule
Exact nonlinear structure with a computable parameter-linear regression.
& The supplied direct, inverse, or integrated relation is retained.
  RLS estimates only its independent unknown parameters.
& The estimates are inserted into the specified nonlinear forward model. \\
\addlinespace
Trusted structural restrictions, but no complete functional model.
& The supplied structure is used to construct a reduced model, such as a
  structure-informed linearization, Taylor/Maclaurin expansion, or other
  structure-informed reduced representation.
  Known zero terms, couplings, and coefficient relations can be used
  to reduce the identified parameterization.
& The retained structure is combined with the fitted approximation. \\
\addlinespace
Little reliable functional structure beyond a selected model class.
& The predictive variables and memory are specified.
  A broader candidate representation, such as polynomial, radial basis
  function (RBF), Fourier, or other fixed nonlinear features, is selected.
  RLS identifies its coefficients.
& The fitted nonlinear model is used. \\
\bottomrule
\end{tabularx}
\end{table}

All three choices define a nonlinear predictor $\widehat F_k(\xi,u)$.
When this map is differentiable on the prediction region, we use its
Jacobian-frozen affine approximation in the same constrained PCAC problem.

%% file: identification.tex
\section{Online Identification and Predictor Reconstruction}
\label{sec:c4-identification}

\subsection{Recursive coefficient update}
Both constructions supply a response $b_k\in\mathbb R^{n_b}$ and a
regressor $\Phi_k\in\mathbb R^{n_b\times n_\theta}$ using data available
through $t_k$. Within a fixed parameter regime, write
\begin{equation}
    b_k=\Phi_k\theta+\nu_k,
    \label{eq:c4-regression}
\end{equation}
where $\nu_k$ contains the applicable measurement, approximation, and
model-discrepancy terms. For the direct model, $n_b=p$. For
\eqref{eq:c4-integrated-data}, $n_b=1$.

We apply the same fixed, specified row scaling to $b_k$ and $\Phi_k$.
The identification equations below use these scaled
quantities. Let $\hat\theta^{\mathrm{rls}}_0\in\mathbb R^{n_\theta}$ and
$P_0=P_0^\top\succ0$ be the initial values. At step $k\geq1$, define
\begin{equation}
    \varepsilon_k=b_k-\Phi_k\hat\theta^{\mathrm{rls}}_{k-1}.
    \label{eq:c4-residual}
\end{equation}
For a forgetting factor $\lambda_k\in(0,1]$, the update is
\begin{align}
    L_k&=\lambda_k^{-1}P_{k-1},\label{eq:c4-inflation}\\
    P_k&=L_k-L_k\Phi_k^\top
       (I_{n_b}+\Phi_kL_k\Phi_k^\top)^{-1}\Phi_kL_k,
       \label{eq:c4-covariance}\\
    \hat\theta^{\mathrm{rls}}_k
       &=\hat\theta^{\mathrm{rls}}_{k-1}
         +P_k\Phi_k^\top\varepsilon_k.
       \label{eq:c4-rls-update}
\end{align}
These updates are standard RLS operations~\cite{ljung1983theory,Bruce2020VRFRLS}.
Shared unknown coefficients are updated jointly. Separate row-wise
updates are used only when no unknown coefficient is shared across rows.

The choice $\lambda_k=1$ disables forgetting. A constant
$\lambda_k=\lambda<1$ gives fixed forgetting. The residual-triggered rule below adapts the windowed construction
in~\cite{Bruce2020VRFRLS} to the specified regression coordinates.
For residual-triggered variable-rate forgetting, choose $N_f\geq1$, $\sigma_\varepsilon>0$,
$\eta>0$, and $\gamma>0$, and define
\begin{align}
    n_{f,k}&=\min\{k,N_f\},\qquad
    E_k=\left(
       \frac{1}{n_{f,k}\sigma_\varepsilon^2}
       \sum_{j=k-n_{f,k}+1}^{k}\|\varepsilon_j\|_2^2
       \right)^{1/2},\label{eq:c4-residual-window}\\
    \lambda_k&=
       \begin{cases}
       \bigl(1+\eta\min\{E_k,\gamma\}\bigr)^{-1},&E_k>1,\\
       1,&E_k\leq1.
       \end{cases}
       \label{eq:c4-vrf}
\end{align}
The residual scale $\sigma_\varepsilon$ is specified in these fixed
regression coordinates. Forgetting changes the influence of past data; it does not change the
regressor or generate excitation.

\subsection{Reconstructed nonlinear predictor}
Let $\mathcal T\subseteq\mathbb R^{n_\theta}$ be the specified admissible
parameter set, and let $\Pi_{\mathcal T}$ be a specified map into that set.
Define
\begin{equation}
    \hat\theta_k=\Pi_{\mathcal T}
        (\hat\theta^{\mathrm{rls}}_k),\qquad
    \widehat F_k(\xi,u)=F(\xi,u;\hat\theta_k).
    \label{eq:c4-reconstruction}
\end{equation}
We use the identity map when no admissibility restriction is required.
For \eqref{eq:c4-physical-forward}, a specified set
$\mathcal T=\{(J,d):J\geq J_{\min},\ d\geq0\}$ with $J_{\min}>0$
retains the prescribed parameter signs and excludes a zero denominator.
The unconstrained estimate remains the RLS state, whereas the mapped
estimate is used for prediction.

The nonlinear one-step output prediction is
\begin{equation}
    \widehat y_{k+1|k}=C\widehat F_k(\xi_k,u_k).
    \label{eq:c4-output-prediction}
\end{equation}
The estimate is updated before this prediction is constructed.
The map $\widehat F_k$ contains both the retained known structure and the
current estimates of its unknown coefficients.

Parameter linearity does not establish parameter convergence.
Measured variables may enter both $b_k$ and $\Phi_k$, so $\nu_k$ need
not be independent of $\Phi_k$~\cite{LaiBernstein2024GFRLS}. Changing dynamics and incomplete model
classes turn identification into a tracking or approximation problem.
These effects do not alter the definition of the reconstructed predictor.

%% file: predictive_control.tex
\section{Local Affine Prediction and Constrained PCAC}
\label{sec:c4-control}

\subsection{Local affine predictor}
Assume that $\widehat F_k$ is differentiable on the region used for
prediction. At the measured point $(\xi_k,u_k)$, define
\begin{equation}
    A_k=\left.\frac{\partial\widehat F_k}{\partial\xi}
           \right|_{(\xi_k,u_k)},\qquad
    B_k=\left.\frac{\partial\widehat F_k}{\partial u}
           \right|_{(\xi_k,u_k)},
    \label{eq:c4-jacobians}
\end{equation}
where $A_k\in\mathbb R^{n\times n}$ and $B_k\in\mathbb R^{n\times m}$.
The estimate $\hat\theta_k$ remains fixed during differentiation. Define
\begin{equation}
    c_k=\widehat F_k(\xi_k,u_k)-A_k\xi_k-B_ku_k.
    \label{eq:c4-affine-offset}
\end{equation}
The local affine map is
\begin{equation}
    \widehat F_k^{\mathrm a}(\xi,u)=c_k+A_k\xi+B_ku.
    \label{eq:c4-affine-map}
\end{equation}
This affine map matches $\widehat F_k$ and its Jacobian at $(\xi_k,u_k)$.
No equilibrium condition is used. In a Taylor-informed model,
$z_\circ$ remains the fixed polynomial expansion point;
$(\xi_k,u_k)$ is the current Jacobian evaluation point.

For \eqref{eq:c4-sine-plant}, with $\xi_k=y_k$, the coefficients are
\begin{equation}
    A_k=\hat a_k,\qquad B_k=\hat b_k\cos u_k,\qquad
    c_k=\hat d_k+\hat b_k(\sin u_k-u_k\cos u_k).
    \label{eq:c4-sine-affine}
\end{equation}
Thus, the known nonlinear function enters both the input derivative and
the affine offset. For a history-based predictor, the linear shift rows
in \eqref{eq:c4-history-map} are retained exactly by this construction.

\subsection{Frozen finite-horizon prediction}
We choose $N\geq2$ and initialize the prediction horizon with
\begin{equation}
    \xi_{0|k}=\xi_k,\qquad u_{0|k}=u_k.
    \label{eq:c4-horizon-initialization}
\end{equation}
For $i=0,\ldots,N-1$, the prediction constraints are
\begin{equation}
    \xi_{i+1|k}=c_k+A_k\xi_{i|k}+B_ku_{i|k},\qquad
    y_{i+1|k}=C\xi_{i+1|k}.
    \label{eq:c4-horizon}
\end{equation}
The matrices $A_k$, $B_k$, and the vector $c_k$ remain fixed over this
horizon and are recomputed at the next sampling instant.

Local linearization defines the affine approximation. Freezing it over
the horizon is a separate computational choice that preserves affine
prediction constraints. The resulting predictions need not match
recursive propagation of $\widehat F_k$ away from the current point.
If the complete reconstructed map is affine, the approximation reproduces
that map. A degree-one unknown part alone does not give an affine model
when a retained known part is nonlinear.

Local-linearization GPC already uses an identified nonlinear model to construct
a prediction model at the current operating point.
N{\o}rgaard et al.~\cite{NorgaardEtAl1996IntelligentPredictive} derive local
derivatives and a linearization offset from a neural NARX model for GPC.
Gonz{\'a}lez Cisneros and Werner~\cite{GonzalezCisnerosWerner2020QuasiLPVNMPC} formulate a
quasi-LPV NMPC procedure that solves a sequence of QPs at each sampling instant.
Here, one current-point affine model is retained over each prediction horizon,
without inner updates of the horizon model. Section~\ref{sec:num-diagnostics} evaluates the resulting
affine-freezing error separately from nonlinear-model error.

\subsection{Constrained quadratic program}
Let $r_{i|k}\in\mathbb R^p$ be the supplied reference over the horizon.
Define
\begin{equation}
    \Delta u_{i|k}=u_{i|k}-u_{i-1|k},\qquad i=1,\ldots,N-1.
    \label{eq:c4-input-increment}
\end{equation}
We specify input bounds $H_uu\leq h_u$ and input-increment bounds
$H_\Delta\Delta u\leq h_\Delta$, and choose $u_0$ within the input bounds.
Optional output bounds
$H_yy\leq h_y$ have $n_c$ rows and are softened with
$\epsilon_{i|k}\in\mathbb R^{n_c}_{\geq0}$.

Let $Q_i\in\mathbb R^{p\times p}$ satisfy $Q_i\succeq0$,
$R\in\mathbb R^{m\times m}$ satisfy $R\succ0$, and
$S_\epsilon\in\mathbb R^{n_c\times n_c}$ satisfy $S_\epsilon\succ0$.
Define $\|v\|_M^2=v^\top Mv$. The controller solves
\begin{equation}
\begin{aligned}
 \underset{\substack{\xi_{1|k},\ldots,\xi_{N|k},\\
                       u_{1|k},\ldots,u_{N-1|k},\\
                       \epsilon_{1|k},\ldots,\epsilon_{N|k}}}
            {\operatorname{minimize}}
 &\quad \frac12\sum_{i=1}^{N}\|C\xi_{i|k}-r_{i|k}\|_{Q_i}^2
       +\frac12\sum_{i=1}^{N-1}\|\Delta u_{i|k}\|_R^2 +\frac12\sum_{i=1}^{N}\|\epsilon_{i|k}\|_{S_\epsilon}^2\\
 \text{subject to}
 &\begin{alignedat}[t]{2}
  &\quad \eqref{eq:c4-horizon-initialization},\ \eqref{eq:c4-horizon},\\
  &\quad H_u u_{i|k}\leq h_u,
                 &\qquad &i=1,\ldots,N-1,\\
  &\quad H_\Delta\Delta u_{i|k}\leq h_\Delta,
                 &\qquad &i=1,\ldots,N-1,\\
  &\quad H_yC\xi_{i|k}\leq h_y+\epsilon_{i|k},
                 &\qquad &i=1,\ldots,N,\\
  &\quad \epsilon_{i|k}\geq0,
                 &\qquad &i=1,\ldots,N.
 \end{alignedat}
\end{aligned}
\label{eq:c4-qp}
\end{equation}
The terminal tracking weight is included in $Q_N$.
When output bounds are omitted, their constraints, the slack variables,
and the slack cost are also omitted.
The input $u_{0|k}$ is fixed; only $u_{1|k},\ldots,u_{N-1|k}$ are optimized.

The objective is quadratic and the constraints are affine.
We use a standard QP solver and apply
\begin{equation}
    u_{k+1}=u_{1|k}^{\star}.
    \label{eq:c4-applied-input}
\end{equation}
This is the PCAC predictive-control step~\cite{nguyen2021predictive,rawlings2020model}.
The QP optimizes the frozen-affine predictor, not the nonlinear plant.
Predicted output bounds do not certify actual plant output bounds.

For the bounds used in the numerical studies, a feasible QP point can
be constructed whenever the committed input satisfies the hard input
bounds and the prediction state and affine coefficients are finite.
Set $u_{i|k}=u_k$ for $i=1,\ldots,N-1$. The resulting zero input
increments satisfy the implemented increment bounds. When output
constraints are included, propagate the affine predictor and set
$\epsilon_{i|k}=\max\{0,H_yC\xi_{i|k}-h_y\}$ componentwise.
The output slacks have no upper bound, so this choice satisfies the
softened constraints. This construction does not guarantee finite model
evaluations at every later step, successful numerical solution,
closed-loop stability, or satisfaction of actual plant-output bounds.

At each step, the controller processes the completed measurements, updates
the unknown coefficients, reconstructs $\widehat F_k$, evaluates $A_k$, $B_k$, and $c_k$,
and solves \eqref{eq:c4-qp}. Structural information determines what is
retained and what is identified. Jacobian freezing determines how the
reconstructed nonlinear model is used for predictive optimization.

%% file: numerical_study.tex
\section{Numerical Studies}
\label{sec:num}

\input{study_design.tex}

\subsection{Structural Information and Model Complexity}
\label{sec:num-complexity}

This subsection addresses the first two questions in
Table~\ref{tab:num-study-map}.
The first comparison changes supplied information at fixed or explicitly
varied parameter dimension. The second changes approximation order and
operating range for a non-polynomial plant.

\input{structural_information.tex}

\input{operating_region.tex}

\subsection{Adaptation to Plant Changes}
\label{sec:num-changes}

This subsection addresses parameter variation and structural activation
in Table~\ref{tab:num-study-map}.
The dictionaries remain fixed. Parameter variation changes coefficients
of already present plant terms. Structural activation changes the set
of nonzero plant terms. The relation of each structural change to the
available dictionary is specified separately.

\input{parameter_variation.tex}

\input{nonlinear_activation.tex}

\input{outside_model_class.tex}

\input{physical_identification.tex}

\input{diagnostics.tex}

%% file: study_design.tex
\subsection{Study Design and Evaluation Protocol}
\label{sec:num-protocol}

The numerical investigation examines how supplied structural information
and model complexity affect identification, prediction, and closed-loop
control. The comparisons first consider stationary plants, then parameter
variation and changes within and outside the retained model class. A physical
example examines the distinction between identification and forward
reconstruction. The final diagnostics separate nonlinear-model error from
the error introduced by affine freezing.
Table~\ref{tab:num-study-map} summarizes the comparisons.

\begin{table}[htbp]
\centering
\small
\setlength{\tabcolsep}{4pt}
\renewcommand{\arraystretch}{1.08}
\caption{Organization of the numerical studies.}
\label{tab:num-study-map}
\begin{tabularx}{\linewidth}{@{}c>{\raggedright\arraybackslash}p{0.16\linewidth}
  >{\raggedright\arraybackslash}X
  >{\raggedright\arraybackslash}X@{}}
\toprule
Study & Section & Question & Comparison \\
\midrule
1 & \ref{sec:num-structure} & What does supplied structure contribute? &
Function choice and coefficient sharing at controlled dimension \\
\addlinespace
2 & \ref{sec:num-region} & When does a richer model help? &
Polynomial degree, restrictions, and operating range \\
\addlinespace
3 & \ref{sec:num-parameters} & How does forgetting affect adaptation? &
Gain steps and drift under three forgetting rules \\
\addlinespace
4 & \ref{sec:num-activation}--\ref{sec:num-outside} &
Can retained terms accommodate new nonlinear behavior? &
No change, activation of a retained term, and an out-of-class change \\
\addlinespace
5 & \ref{sec:num-physical} & How does physical identification support prediction? &
Integrated identification and direct sampled-map models \\
\addlinespace
6 & \ref{sec:num-diagnostics} & What changes after affine freezing? &
True, learned nonlinear, and frozen-affine forecasts under identical inputs \\
\bottomrule
\end{tabularx}
\end{table}

\paragraph{Common settings}
The studies use two scalar sampled recurrences and one scalar differential
equation, specified below. For each plant,
\begin{equation}
 y_k=x_k+v_k,\qquad \xi_k=y_k,\qquad C=1.
 \label{eq:num-measurement}
\end{equation}
The controller uses the measured state; the true state is used only for
evaluation. All models use the identification and control procedures of
Sections~\ref{sec:c4-identification} and~\ref{sec:c4-control}, including the committed input $u_k$ and the computation
of $u_{k+1}$. The dictionaries remain fixed during each experiment; the
regressors below specify the retained terms and coefficient relations.

The known-model reference ($\mathsf K$) uses the current true parameters,
the same measured state, and the same frozen-affine QP. It has no knowledge
of future plant changes. For Plant C, it uses the same numerical predictor
integrator as the integrated physical model. This reference is diagnostic,
not a performance bound.

The common settings are
\begin{equation}
 \begin{aligned}
 T_s&=0.1\ \mathrm{s},& N&=10,\\
 Q_i&=1\quad(1\leq i<N),& Q_N&=10,\\
 R&=0.05,& S_\epsilon&=10^4I_2.
 \end{aligned}
 \label{eq:num-common-settings}
\end{equation}
A separate sensitivity comparison in Appendix~A
(Table~\ref{tab:ejc_p06_sensitivity}) varies the input-increment penalty,
normalized initial covariance scale, and ini\-tial\-iza\-tion-record length one
at a time. The settings of the primary comparisons remain unchanged.

States and inputs are normalized by one unit of the corresponding plant
variable. Table~\ref{tab:num-bounds} gives the hard input and
input-increment bounds, soft output bounds, and durations. Within each
comparison, the weights, bounds, reference preview
$r_{i|k}=r(t_k+iT_s)$, and initial conditions are identical. Control runs
start from $x_0=u_0=0$, except in the constraint audit. No additional online
probing is applied.

\begin{table}[htbp]
\centering
\small
\caption{Bounds and run durations in each plant's stated units.}
\label{tab:num-bounds}
\begin{tabular}{@{}lrrrr@{}}
\toprule
Plant & $u_{\max}$ & $\Delta u_{\max}$ & $x_{\max}$ & Duration (s) \\
\midrule
A & 1.0 & 0.25 & 1.2 & 120 \\
B & 1.2 & 0.25 & 1.2 & 80 \\
C & 1.0 & 0.10 & 1.2 & 80 \\
B: constraint audit & 1.2 & 0.25 & 0.85 & 10 \\
\bottomrule
\end{tabular}
\end{table}

\paragraph{Identification and evaluation}
Except in the sensitivity comparison, each adaptive model is fitted to the
same separate 200-transition record from the initial plant regime, with
$\lambda_k=1$. Its fitted coefficients
and covariance initialize the control run after the plant is restarted.
The prior--frozen control is the stated exception.
Independent 600-transition records evaluate prediction accuracy after
25, 50, 100, and 200 fitting transitions. These tests use common data, whereas
subsequent closed-loop identification data can differ because the
controllers apply different inputs.

Fixed row and column scalings are applied before RLS. In the normalized
parameter coordinates, the initial covariance is
$P_{\beta,M,\mathrm{init}}=(100/n_{\theta,M})I_{n_{\theta,M}}$,
where $n_{\theta,M}$ is the number of independent coefficients.
This convention equalizes average prior regression variance on the
specified calibration grids, not physical parameter uncertainty across
different identification relations. The initial models are specified with
each study. \ref{app:num-protocol} gives the input records,
scaling definitions, and calibration grids.

The primary comparisons are noise-free. For the stationary study of
supplied structure, 40 paired trials use measurement noise
$v_k\sim\mathcal N(0,0.01^2)$ and $w_k=0$. The selected gain-step
comparison uses 40 paired trials with the same measurement noise and
$w_k\sim\mathcal N(0,0.005^2)$. The sequences are mutually and temporally
independent and are paired across models by sample index. The physical
study remains deterministic. The trial selections and independent
initialization and control noise streams are specified in the appendix.

\paragraph{Performance measures and checks}
The principal measures are tracking and nonlinear prediction
root-mean-square errors (RMSEs), input RMS, and mean absolute input increment.
Tracking is scored from
$x_k-r_k$ over fixed evaluation windows. One-step prediction is scored
against the next state, before that transition is used for identification.
Parameter error is evaluated only for models with an exact reference
parameter vector. Section~\ref{sec:num-diagnostics} defines the nonlinear
and affine forecast errors. The appendix gives the complete metric
definitions, evaluation windows, recovery criteria, and parameter scales.

Numerical checks cover reconstruction, the RLS update, input timing,
predictor derivatives, integration, and QP residuals. Predicted slack and
actual bound violations are recorded separately. Rejected identification
updates retain the last valid estimate and covariance; rejected control
computations hold the input. Runs terminate if the state becomes nonfinite
or its magnitude exceeds five state units. All attempted trials and failures
are reported; incomplete trajectories are not scored as complete runs.
Repeated-trial summaries use medians, interquartile ranges, and paired
differences, with uncertainty calculations specified in the appendix.

%% file: structural_information.tex
\subsubsection{Supplied structure at fixed parameter dimension}
\label{sec:num-structure}

Consider Plant A,
\begin{equation}
 x_{k+1}=d_{\mathrm A}+a_{\mathrm A}x_k
 +b_{\mathrm A,k}(1+c_{\mathrm A}x_k)u_k
 +h_kx_k^2+\rho_k\sin(2x_k)+w_k.
 \label{eq:num-plant-a}
\end{equation}
The state and input are dimensionless, and the recurrence defines the
sampled plant directly. For the stationary comparison, we set
\begin{equation}
 (d_{\mathrm A},a_{\mathrm A},c_{\mathrm A})=(0.03,0.65,0.35),
 \qquad
 b_{\mathrm A,k}=b_{\mathrm A}=0.45,
 \qquad h_k=\rho_k=0.
 \label{eq:num-a-nominal}
\end{equation}
Table~\ref{tab:num-a-models} lists the compared adaptive models.
The shared model ($\mathsf S$) enforces the correct coupling relation with
$c_{\mathrm A}=0.35$. The incorrectly shared model ($\mathsf W$) enforces the opposite sign.
The relaxed model ($\mathsf R$) and the complete quadratic model
($\mathsf P_2$) fit the relevant coefficients
independently.
No adaptive model receives the true values of
$d_{\mathrm A}$, $a_{\mathrm A}$, or $b_{\mathrm A,k}$.

\begin{table}[!htbp]
\centering
\small
\caption{Direct models for stationary Plant A.}
\label{tab:num-a-models}
\begin{tabularx}{\linewidth}{@{}lXr@{}}
\toprule
Model & Complete regressor row $\Psi_M(\xi,u)$ & $n_{\theta,M}$ \\
\midrule
Affine ($\mathsf A$) &
$[1,\xi,u]$ & 3 \\
Shared ($\mathsf S$) &
$[1,\xi,(1+c_{\mathrm A}\xi)u]$ & 3 \\
Incorrectly shared ($\mathsf W$) &
$[1,\xi,(1-c_{\mathrm A}\xi)u]$ & 3 \\
Relaxed ($\mathsf R$) &
$[1,\xi,u,\xi u]$ & 4 \\
Complete quadratic ($\mathsf P_2$) &
$[1,\xi,u,\xi^2,\xi u,u^2]$ & 6 \\
\bottomrule
\end{tabularx}
\end{table}

The exact stationary coefficient vectors are
\begin{align}
 \theta_{\mathsf S}^{\star}
 &=(d_{\mathrm A},a_{\mathrm A},b_{\mathrm A})^\top,
 \nonumber\\
 \theta_{\mathsf R}^{\star}
 &=(d_{\mathrm A},a_{\mathrm A},b_{\mathrm A},
    b_{\mathrm A}c_{\mathrm A})^\top,
 \nonumber\\
 \theta_{\mathsf P_2}^{\star}
 &=(d_{\mathrm A},a_{\mathrm A},b_{\mathrm A},0,
    b_{\mathrm A}c_{\mathrm A},0)^\top.
 \label{eq:num-a-targets}
\end{align}
The affine and incorrectly shared models do not contain the nominal plant on an
open state--input region. No exact parameter target is assigned to them.
The complete map of the affine model is affine.

We use $\lambda_k=1$ and include the known-model reference.
Comparing the shared model with the affine and incorrectly shared models
changes functional structure or restriction correctness at the same
coefficient count. Comparing the shared and relaxed models relaxes sharing
within the correct set of terms. Comparing the relaxed and complete quadratic
models adds terms whose true coefficients are zero.

Before fitting, we initialize the shared, relaxed, and complete quadratic models
with the common map $0.5\xi+0.25(1+c_{\mathrm A}\xi)u$.
The affine model is initialized with $0.5\xi+0.25u$ and the incorrectly shared
model with $0.5\xi+0.25(1-c_{\mathrm A}\xi)u$.
All five initial maps have the same value and Jacobian at the origin.

Let $t$ be measured in seconds and define
\begin{equation}
 \begin{aligned}
 r_{\mathrm A}(t)&=r_j,\qquad
 j=\min\{5,\lfloor t/20\rfloor\},\\
 (r_0,\ldots,r_5)&=(0.6,-0.5,0.75,-0.4,0.6,0).
 \end{aligned}
 \label{eq:num-a-reference}
\end{equation}
The equilibrium input is
$u_{\mathrm{eq}}(r)=((1-a_{\mathrm A})r-d_{\mathrm A})/
[b_{\mathrm A}(1+c_{\mathrm A}r)]$.
The prescribed reference levels lie within the steady-state input authority
for both $b_{\mathrm A}=0.45$ and $0.30$.
This algebraic check does not establish transient feasibility.

Figure~\ref{fig:num-study1} compares held-out prediction and whole-run
tracking for stationary Plant A. The known-model reference performs no
identification; its held-out prediction RMSE is $2.82 \times 10^{-17}$
(numerical roundoff) and is excluded from the logarithmic axis in panel (a).
Some interquartile ranges are narrower than the symbols.

\begin{figure}[!htbp]
\centering
\includegraphics[width=\linewidth]{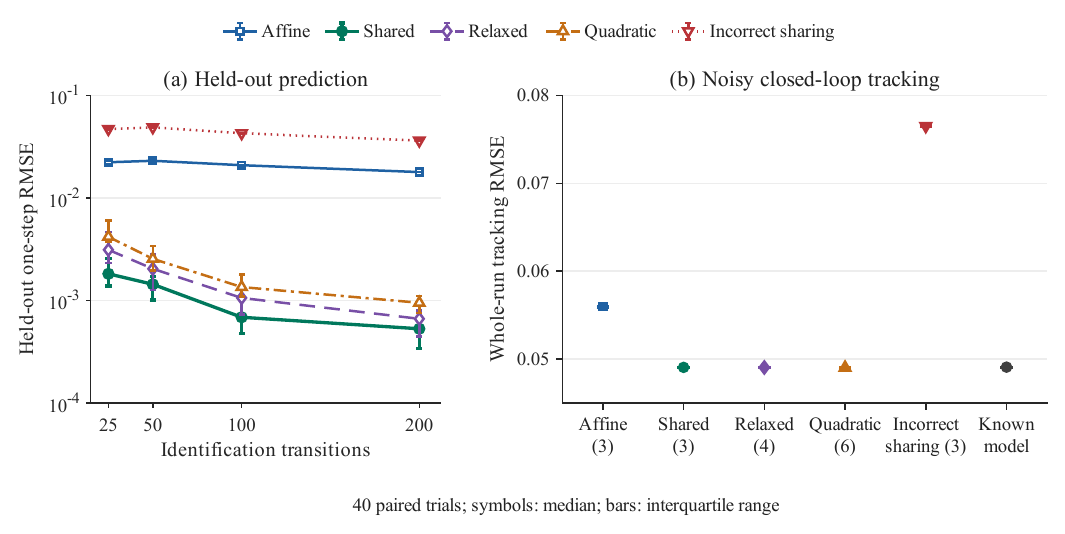}
\input{study1_main_caption.tex}
\label{fig:num-study1}
\end{figure}

The shared model improves prediction without increasing the
number of estimated coefficients relative to the affine model.
After 200 initialization transitions, the median one-step error on the
independent noise-free 600-transition evaluation record is $5.30\times10^{-4}$ for
the shared model and $1.79\times10^{-2}$ for the affine model across the 40 noisy
fits. The corresponding values are $6.63\times10^{-4}$ for the relaxed
model and $9.50\times10^{-4}$ for the complete quadratic model.
Thus, supplying the correct coefficient relation
reduces the median prediction error under the specified data and
initialization conditions. The incorrectly shared model
has a larger error than the affine model.

The prediction differences are not proportional to the differences in
closed-loop tracking. The whole-run scores cover $0 \leq t < 120$~s and
retain all reference transitions. The median whole-run tracking RMSE decreases from
$0.05597$ for the affine model to $0.04906$ for the shared model, a reduction of
approximately $12.3\%$. The values for the relaxed model, the complete quadratic model, and
the known-model reference are $0.04905$, $0.04903$, and $0.04909$,
respectively. The shared, relaxed, and complete quadratic models have close whole-run
tracking medians in absolute terms. Additional coefficients reduce this
tracking median slightly but increase the held-out prediction error in
this comparison. Table~\ref{tab:ejcArchivePaired} reports the paired tracking differences and
their intervals.

Table~\ref{tab:num-structure-results} summarizes the reported medians.

\begin{table}[!htbp]
\centering
\small
\caption{Medians for stationary Plant A. One-step errors use the independent
noise-free evaluation record after 200 noisy initialization transitions;
tracking RMSE uses the whole control run. A dash denotes a value not tabulated here.}
\label{tab:num-structure-results}
\begin{tabular}{@{}lrr@{}}
\toprule
Model & One-step RMSE & Tracking RMSE \\
\midrule
Affine ($\mathsf A$) & $1.79\times10^{-2}$ & $0.05597$ \\
Shared ($\mathsf S$) & $5.30\times10^{-4}$ & $0.04906$ \\
Relaxed ($\mathsf R$) & $6.63\times10^{-4}$ & $0.04905$ \\
Complete quadratic ($\mathsf P_2$) & $9.50\times10^{-4}$ & $0.04903$ \\
Known-model reference ($\mathsf K$) & --- & $0.04909$ \\
\bottomrule
\end{tabular}
\end{table}

%% file: study1_main_caption.tex
\caption{Study 1: structural information, identification, and closed-loop tracking.
(a) Held-out one-step prediction RMSE after 25, 50, 100, and 200 identification
transitions. (b) Whole-run tracking RMSE for 40 paired noisy closed-loop trials.
Symbols denote medians and bars interquartile ranges. Parentheses give the
number of estimated coefficients. The known-model prediction error in (a)
is at numerical roundoff.}

%% file: operating_region.tex
\subsubsection{Approximation order and operating region}
\label{sec:num-region}

Consider Plant B,
\begin{equation}
 x_{k+1}=0.03+0.65x_k+0.40\sin u_k+w_k.
 \label{eq:num-plant-b}
\end{equation}
This dimensionless recurrence is non-polynomial in the input.
We compare the known-model reference with the direct models
\begin{align}
 \Psi_{\mathsf A}&=[1,\xi,u],
 &n_{\theta,\mathsf A}&=3,
 \nonumber\\
 \Psi_{\mathsf E}&=[1,\xi,\sin u],
 &n_{\theta,\mathsf E}&=3,
 \nonumber\\
 \Psi_{\mathsf O_3}&=[1,\xi,u,u^3],
 &n_{\theta,\mathsf O_3}&=4,
 \nonumber\\
 \Psi_{\mathsf O_5}&=[1,\xi,u,u^3,u^5],
 &n_{\theta,\mathsf O_5}&=5,
 \nonumber\\
 \Psi_{\mathsf P_3}
 &=[1,\xi,u,\xi^2,\xi u,u^2,\xi^3,\xi^2u,\xi u^2,u^3],
 &n_{\theta,\mathsf P_3}&=10.
 \label{eq:num-b-models}
\end{align}
The exact-sine model ($\mathsf E$) receives the sine function but not its coefficient.
Its exact coefficient vector is $(0.03,0.65,0.40)^\top$.
The odd cubic ($\mathsf O_3$) and odd quintic ($\mathsf O_5$) models receive affine state dependence and
odd input dependence. Their expansion point is the origin and their
odd-power coefficients are independent. Sine-series coefficient ratios
are not supplied. The complete cubic model ($\mathsf P_3$) removes these restrictions within
the complete cubic dictionary.
None of the polynomial models contains the sine dependence exactly on a
nontrivial input interval.

We use $\lambda_k=1$ and, before fitting, initialize all polynomial models with
$0.5\xi+0.25u$ and the exact-sine model with $0.5\xi+0.25\sin u$.
For separate control runs, define
\begin{equation}
 r_{\mathrm B}(t)=0.08+A_r\sin(2\pi\,0.05t),\qquad
 A_r\in\{0.2,0.6,1.0\}.
 \label{eq:num-b-reference}
\end{equation}
The control runs use the same initialization record and restart state for every amplitude.
Only the reference amplitude changes between these control runs.
The same fitted models are evaluated on the three independent input-range
records specified in Section~\ref{sec:num-protocol}.
Reference amplitude and evaluation input range are separate variables.

Comparing the odd cubic and complete cubic models changes restrictions
at the same maximum degree. Comparing the odd cubic and odd quintic models
increases degree under the same supplied information. The comparison with the exact-sine model
separates polynomial approximation from fitting the exact function.
The input derivative satisfies
$0.40\cos u\geq0.40\cos(1.2)>0$ on the actuator interval.
The realized state and applied-input ranges on \(20\leq t<80\) s are
reported in Table~\ref{tab:ejcArchiveRanges}. Each row gives the minimum
and maximum across the six model runs at that reference amplitude;
these are realized ranges, not prescribed bounds or ranges attained
by every individual controller.

\begin{table}[tbp]
\centering
\caption{Realized ranges in Study 2 over \(20\leq t<80\) s, pooled across
the six model runs at each reference amplitude. State and input are
dimensionless. The final unused committed input is excluded.}
\label{tab:ejcArchiveRanges}
\small
\setlength{\tabcolsep}{4pt}
\begin{tabular}{crrrr}
\hline
\(A_r\) & Minimum state & Maximum state & Minimum input & Maximum input\\
\hline
0.2 & -0.13899 & 0.29442 & -0.19841 & 0.18416 \\
0.6 & -0.53729 & 0.69589 & -0.57811 & 0.56478 \\
1 & -0.92033 & 1.08025 & -1.08144 & 1.06039 \\
\hline
\end{tabular}
\end{table}

\paragraph{Results}
Figure~\ref{fig:study2-main} separates prediction after the common
200-transition fit from tracking during subsequent online adaptation.
The results are from the deterministic, noise-free comparison.
Each fitted predictor is evaluated without refitting on three independent
600-transition records. Their input ranges $[-L,L]$ remain within the
initialization input envelope; larger-input extrapolation is not tested.
The control runs instead vary the reference amplitude, with tracking
summarized over the prescribed interval $20\leq t<80$~s.

The exact-sine model has the lowest held-out error among the adaptive
models at all three evaluation ranges. At $L=1.1$, its RMSE is
$2.095\times10^{-5}$, compared with $1.659\times10^{-2}$ for the affine
model, although both estimate three coefficients. The odd cubic gives
$2.053\times10^{-4}$ without retaining the exact sine function.
Neither adding the fifth-order term nor using the complete cubic
dictionary improves this initial prediction score. The fitted polynomial
errors do not increase monotonically with the evaluation input range.

The closed-loop comparison gives a different ordering among the polynomial
models. The odd quintic has lower post-transient tracking error than the
odd cubic at all three reference amplitudes. At $A_r=1.0$, the error
falls from $4.356\times10^{-4}$ to $2.368\times10^{-4}$, a reduction of
$45.6\%$, but an absolute difference of only $1.988\times10^{-4}$ state
units. The complete cubic, with ten estimated coefficients, does not
provide a consistent tracking improvement over the four-coefficient odd
cubic. Table~\ref{tab:study2-summary} also reports the whole-run scores,
which differ much less among the nonlinear models.
Initial prediction and later tracking use different estimates and data;
their rankings need not coincide.

At the largest reference amplitude, the exact-sine model and known-model
reference give closely comparable tracking in absolute terms. The known-model
reference uses the same frozen-affine predictor and QP and does not provide a theoretical
lower bound on tracking error. The small difference does not establish
an advantage of parameter estimation itself.

All 18 amplitude runs completed without rejected identification updates
or control computations. Re-evaluation of the recorded 50-transition
Gram matrices found numerical rank deficiency at the declared relative
threshold $10^{-10}$ for the complete cubic at
all three amplitudes, and for the odd quintic at the smallest amplitude.
\ref{app:num-protocol} specifies the rank criterion and reporting convention.
These finite-window diagnostics concern recent data, not the accumulated
RLS information, and do not establish the cause of a performance difference.

Under these noise-free conditions and the chosen initialization, useful
nonlinear functions and restrictions matter more than coefficient count
alone. Additional degree can improve sustained tracking without improving
the initial prediction score.

\begin{figure}[tb]
\centering
\includegraphics[width=\linewidth]{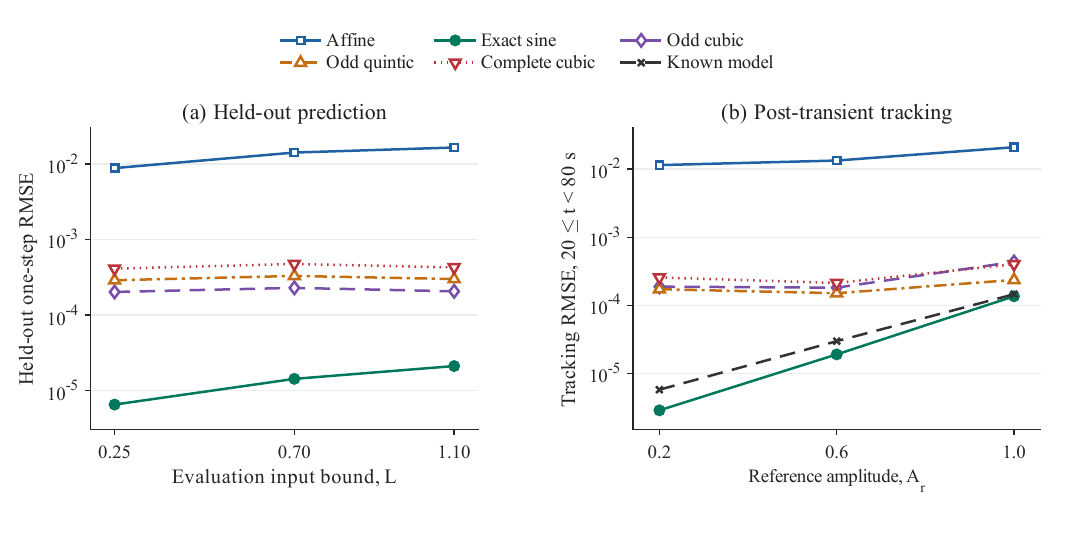}
\caption{Study 2: approximation order and operating region.
(a) Held-out one-step prediction RMSE after 200 identification transitions,
versus evaluation input bound $L$.
(b) Tracking RMSE on $20\leq t<80$~s, versus reference amplitude $A_r$.
Results are deterministic and noise-free.
The known-model prediction error is at numerical roundoff and is omitted
from (a).}
\label{fig:study2-main}
\end{figure}

\begin{table}[tb]
\centering
\small
\setlength{\tabcolsep}{3.5pt}
\renewcommand{\arraystretch}{1.12}
\caption{Study 2 at $A_r=1.0$.
The coefficient count is $n_\theta$.
Input RMS and mean absolute increment use $20\leq t<80$~s.}
\label{tab:study2-summary}
\begin{tabularx}{\linewidth}{@{}>{\raggedright\arraybackslash}Xrrrrr@{}}
\toprule
Model & $n_\theta$ & \multicolumn{2}{c}{Tracking RMSE} &
\shortstack{Input\\RMS} & \shortstack{Mean\\$|\Delta u|$} \\
\cmidrule(lr){3-4}
& & $0\leq t<80$ & $20\leq t<80$ & & \\
\midrule
Affine & 3 & $2.174\times10^{-2}$ & $2.098\times10^{-2}$ & 0.7092 & 0.02042 \\
Exact sine & 3 & $4.126\times10^{-3}$ & $1.361\times10^{-4}$ & 0.7124 & 0.02140 \\
Odd cubic & 4 & $4.148\times10^{-3}$ & $4.356\times10^{-4}$ & 0.7123 & 0.02141 \\
Odd quintic & 5 & $4.136\times10^{-3}$ & $2.368\times10^{-4}$ & 0.7123 & 0.02140 \\
Complete cubic & 10 & $4.148\times10^{-3}$ & $4.023\times10^{-4}$ & 0.7124 & 0.02142 \\
Known-model reference & 0 & $4.127\times10^{-3}$ & $1.459\times10^{-4}$ & 0.7124 & 0.02140 \\
\bottomrule
\end{tabularx}
\end{table}

%% file: parameter_variation.tex
\subsubsection{Parameter variation}
\label{sec:num-parameters}

We use Plant A with the nominal $d_{\mathrm A}$, $a_{\mathrm A}$, and
$c_{\mathrm A}$, and set $h_k=\rho_k=0$.
For the abrupt change, define
\begin{equation}
 b_{\mathrm A,k}=0.45-0.15\,\mathbf 1_{\{t_k\geq50\}}.
 \label{eq:num-gain-step}
\end{equation}
In a separate experiment, define the drift
\begin{equation}
 b_{\mathrm A,k}=0.45-0.15
 \min\{1,\max\{0,(t_k-50)/40\}\}.
 \label{eq:num-gain-drift}
\end{equation}
Both changes preserve the shared representation.
Both experiments use reference~\eqref{eq:num-a-reference} and the same initialization fit.

Each of the affine and shared models is paired with no forgetting, fixed forgetting
$\lambda_k=0.99$, and the residual-triggered variable-rate forgetting
rule in~\eqref{eq:c4-vrf}. The parameters of this rule are
\begin{equation}
 N_f=10,\qquad \sigma_\varepsilon=0.03,\qquad
 \eta=0.02,\qquad \gamma=10.
 \label{eq:num-forgetting-settings}
\end{equation}
The residual scale is specified after the fixed row scaling.
The online residual window starts empty at the control-run restart and
is not cleared at the plant change.

The comparison also includes the known-model reference and two fixed shared-model controls.
The \emph{pretrained--frozen} control retains the fitted model after the
common initialization record. The \emph{prior--frozen} control retains
the raw initial model without fitting or online updates.
These controls separate the value of continued adaptation from the value
of the initialization data.

No adaptive controller receives the event time, future parameter values,
or a switch-triggered reset. A change indexed by $k_c$ first affects
$x_{k_c+1}$. The first changed transition reaches RLS at $k_c+1$.
The known-model reference receives the current parameter at $k_c$.
Drift is evaluated over the changing interval $\mathcal W(50,90)$ and the
late interval $\mathcal W(90,100)$ separately; no settling time is assigned.

\paragraph{Results}
The study comprises nine noise-free controllers for each gain schedule
and 40 paired mixed-noise realizations of the five-controller shared-model
comparison for the abrupt change. All 218 runs completed without rejected
identification updates or control computations. Figure~\ref{fig:study3-main}
shows the noise-free shared-model responses.
Table~\ref{tab:study3-deterministic} in \ref{app:num-protocol}
retains the complete deterministic comparison.

Following the abrupt change, both forgetting rules reduce the late-window
tracking error relative to no forgetting or a model frozen after fitting.
Table~\ref{tab:study3-noisy} shows that this result also holds in the
mixed-noise trials. The median late-window RMSE is 0.07244 without
forgetting, 0.02417 with fixed forgetting, and 0.02777 with variable-rate
forgetting, compared with 0.12574 for the pretrained--frozen model.
These scores cover the prescribed window and are not steady-state errors.

{\emergencystretch=2em
The event-window and recovery comparisons give different orderings.
Variable-rate forgetting has a lower event-window RMSE than fixed
forgetting, but this does not imply earlier recovery. In the noise-free
shared-model run, fixed forgetting satisfies the prescribed tolerance
and two-second dwell with a recovery delay of 5.8~s; variable-rate
forgetting does not complete the dwell by 60~s. The event window
also includes anticipation of the reference change at 60~s. In particular,
the fixed-forgetting response has a larger error near that change.
The declared window is retained, but its RMSE should not be interpreted
as an isolated response to the gain step.\par}

In the mixed-noise trials, none of the four fitted shared-model variants
satisfies the complete recovery criterion by 60~s. These 160 outcomes
are right-censored, not assigned a recovery delay of 10~s. The known-model
reference recovers in all 40 trials, with delays of 0.2--0.3~s.
Non-recovery under the specified dwell condition is distinct from an
incomplete run or a failure of the numerical solver.

The variable-rate rule responds differently to the two change schedules.
For the noise-free shared model, forgetting first activates at 50.2~s
after the abrupt change, but only at 75.1~s during the drift that begins
at 50~s. Over the drift interval $50\leq t<90$~s, tracking RMSE is
0.07349 without forgetting, 0.05221 with fixed forgetting, and 0.06579
with variable-rate forgetting. The corresponding late-window scores are
0.09833, 0.03179, and 0.04501. With the specified residual threshold,
fixed forgetting gives lower tracking error for this gradual variation. No general ordering of forgetting rules follows.

Improved tracking does not imply recovery of the individual plant
coefficients. Although only the input gain changes in the plant, all
three shared-model coefficients remain adaptive. In the noise-free
abrupt-change run, the fixed-forgetting gain estimate falls to about
0.087 even though the true gain is 0.30; the estimated state coefficient
simultaneously increases to about 0.917. This compensation, together with
the frequent numerical rank deficiency of recent noise-free regressor
windows, is consistent with limited information for distinguishing
individual coefficients. The finite-window diagnostic does not establish
singularity of the accumulated RLS information or parameter convergence.

The affine and shared models also exchange rankings in the abrupt-event
window. Thus, retaining the correct representation does not guarantee
the smallest tracking score for every transient. The prior--frozen
control has substantially larger errors and violates the soft output
bound in both noise-free schedules, although its hard input bounds
remain satisfied. Overall, this study supports a distinction between
closed-loop readaptation, residual-triggered forgetting, and recovery
of the true coefficients, even when the changing plant remains within
the retained model class.

\begin{figure}[!tbp]
\centering
\includegraphics[width=\linewidth]{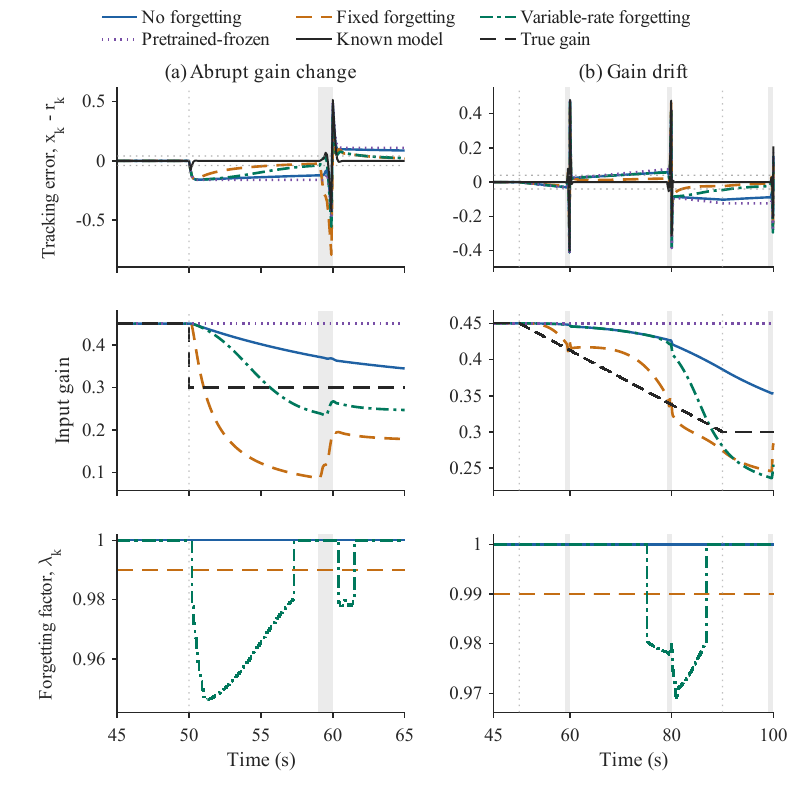}
\caption{Study 3: noise-free shared-model responses to an abrupt gain
change (left) and drift (right). Rows show tracking error, gain estimate,
and forgetting factor. The known-model tracking reference and true gain
are included. Shading marks reference-preview intervals.}
\label{fig:study3-main}
\end{figure}

\begin{table}[tb]
\centering
\small
\setlength{\tabcolsep}{4pt}
\renewcommand{\arraystretch}{1.16}
\caption{Study 3 mixed-noise abrupt change. Tracking entries are medians
[first, third quartiles]. Recovery uses the prescribed two-second dwell.
Each row contains 40 completed runs.}
\label{tab:study3-noisy}
\begin{tabularx}{\linewidth}{@{}>{\raggedright\arraybackslash}Xccc@{}}
\toprule
Controller & \multicolumn{2}{c}{Tracking RMSE} &
\shortstack{Recovered\\by 60~s}\\
\cmidrule(lr){2-3}
& $50\leq t<60$ & $90\leq t<100$ & \\
\midrule
No forgetting &
\shortstack{0.14795\\{\footnotesize [0.14747, 0.14855]}} &
\shortstack{0.07244\\{\footnotesize [0.07181, 0.07283]}} & 0/40 \\
Fixed forgetting &
\shortstack{0.14188\\{\footnotesize [0.13818, 0.14329]}} &
\shortstack{0.02417\\{\footnotesize [0.02345, 0.02464]}} & 0/40 \\
Variable-rate forgetting &
\shortstack{0.12383\\{\footnotesize [0.12332, 0.12420]}} &
\shortstack{0.02777\\{\footnotesize [0.02720, 0.02809]}} & 0/40 \\
Pretrained--frozen &
\shortstack{0.16628\\{\footnotesize [0.16489, 0.16730]}} &
\shortstack{0.12574\\{\footnotesize [0.12505, 0.12719]}} & 0/40 \\
Known-model reference &
\shortstack{0.05064\\{\footnotesize [0.05006, 0.05128]}} &
\shortstack{0.02299\\{\footnotesize [0.02250, 0.02362]}} & 40/40 \\
\bottomrule
\end{tabularx}
\end{table}

%% file: nonlinear_activation.tex
\subsubsection{Activation of an initially unnecessary nonlinear term}
\label{sec:num-activation}

For the structural-change experiments only, we specialize Plant A by setting
$d_{\mathrm A}=c_{\mathrm A}=0$ and $b_{\mathrm A,k}=0.45$.
Thus,
\begin{equation}
 x_{k+1}=0.65x_k+0.45u_k+h_kx_k^2
          +\rho_k\sin(2x_k)+w_k.
 \label{eq:num-linear-start-plant}
\end{equation}
The plant is linear when $h_k=\rho_k=0$.
This specialization removes the nominal $x_ku_k$ interaction; it is not
used in the stationary or gain-variation studies.

The comparison includes the affine and complete quadratic models, the
known-model reference, and the augmented affine model ($\mathsf A_+$).
The augmented affine model adds the single basis function $\xi^2$ to the
affine dictionary, giving
\begin{equation}
 \Psi_{\mathsf A_+}(\xi,u)=[1,\xi,u,\xi^2],\qquad
 n_{\theta,\mathsf A_+}=4.
 \label{eq:num-augmented-affine}
\end{equation}
Before fitting, all adaptive models are initialized with $0.5\xi+0.25u$.
The constant coefficient remains an unknown even though its true value
is zero. We fit the adaptive models to a new common initialization record from the
linear regime. Every adaptive model uses the variable-rate forgetting
settings in \eqref{eq:num-forgetting-settings}, and all cases retain the
same reference, bounds, and control weights.

Define the no-change control and the represented change by
\begin{align}
 \text{No change:}\quad
 &h_k=0,\qquad \rho_k=0,
 \nonumber\\
 \text{Represented change:}\quad
 &h_k=h_\star\mathbf 1_{\{t_k\geq50\}},\qquad
 \rho_k=0,\qquad h_\star=0.08.
 \label{eq:num-structural-cases}
\end{align}
For each model, the same fitted initialization state is reused in both cases.
The no-change control retains the same reference and reporting windows.

In the no-change and represented-change cases, the exact vectors are
\begin{equation}
 \begin{aligned}
 \theta_{\mathsf A_+,k}^{\star}
 &=(0,0.65,0.45,h_k)^\top,\\
 \theta_{\mathsf P_2,k}^{\star}
 &=(0,0.65,0.45,h_k,0,0)^\top.
 \end{aligned}
 \label{eq:num-activation-targets}
\end{equation}
Thus the true $\xi^2$ coefficient changes from zero to $h_\star$ in the
represented-change case. The other extra coefficients remain zero.
The affine model is exact only in the linear regime.

Every retained coefficient is estimated from the beginning;
none is released or inserted at the switch.
An initially inactive term means a zero true coefficient, not a regressor
that evaluates to zero or is excluded from RLS.
The resulting adaptation occurs within a preallocated representation, not through online
structure discovery.

The no-change control measures the consequence of carrying unnecessary
terms. The represented change tests whether their availability becomes
useful after the plant becomes nonlinear.
The estimated $\xi^2$ coefficient is recorded together with common-query
prediction and tracking errors.

Online model snapshots are evaluated every second on the fixed grid
$x\in[-1,1]$, $u\in[-0.7,0.7]$, with 21 points per axis, against the
current noise-free plant map. The no-change control uses the same grid.
An estimated nonzero coefficient alone does not establish successful
identification or a control benefit.

\begingroup
\emergencystretch=1em
\paragraph{No-change and represented-change results}
The deterministic comparison comprises four models under each of the three
plant schedules. All 12 runs completed without rejected identification
updates, control fallbacks, or violations of the prescribed plant-output
and hard input or input-increment bounds. Table~\ref{tab:study4-results} summarizes the
tracking, common-grid prediction, and recovery measures.\par
\endgroup

Without a plant change, all four controllers have a whole-run tracking
RMSE of approximately 0.04958. After the common 200-transition
initialization record, held-out prediction RMSE is
$3.31\times10^{-5}$ for the affine model,
$4.41\times10^{-5}$ for the augmented affine model, and
$6.66\times10^{-5}$ for the complete quadratic model.
The additional coefficients therefore increase finite-data prediction
error in this case, but produce little difference in the tracking score.

After the $x^2$ term is introduced, its estimates become appreciably
nonzero in both enriched models
(Figure~\ref{fig:study4-main}, top left).
The augmented affine and complete quadratic models satisfy the recovery
criterion after 6.8 and 5.8~s, respectively; the affine model is
right-censored at 60~s. Their late-window common-grid prediction RMSEs
are 0.02464 and 0.02588, compared with 0.04275 for the affine model.
Thus, retaining the added term is associated with earlier recovery and
lower regional prediction error in this comparison.

These improvements do not imply accurate coefficient recovery or a
uniform tracking advantage. At the last stored estimate, $t=119.9$~s,
the estimated $\xi^2$ coefficients are 0.0512 and 0.1099, compared with
the true value 0.08. The complete quadratic estimate reaches 0.1904
during the initial response. The affine model also has lower
late-window tracking RMSE than either enriched model.
The event and late windows include anticipation of reference changes at
60 and 100~s, respectively. They are retained as specified, but their
tracking scores are not pure measures of settled tracking or recovery
speed.

\begin{figure}[!htbp]
\centering
\includegraphics[width=\linewidth]{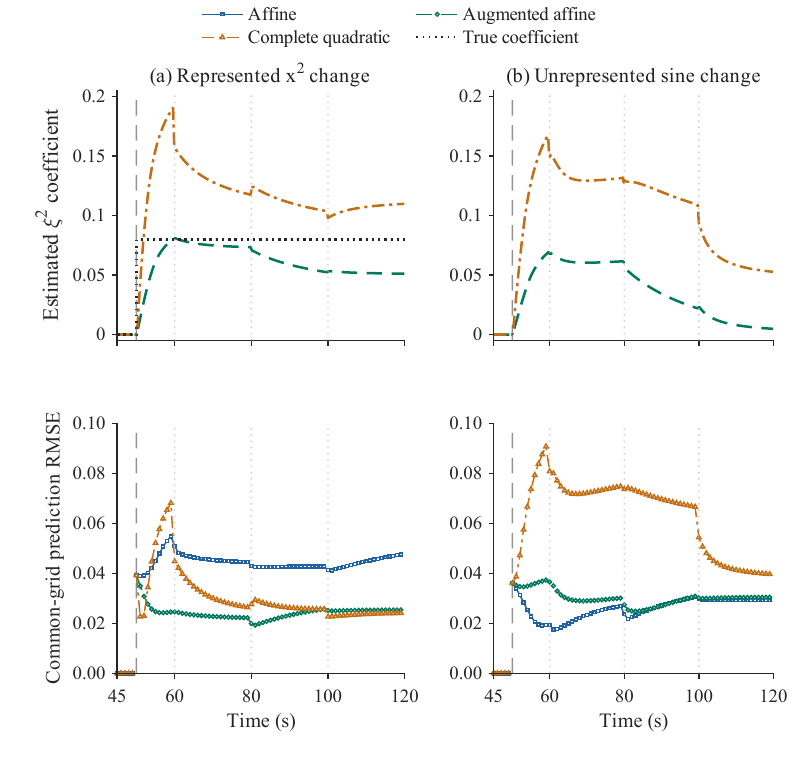}
\caption{Study 4: represented $x^2$ change (left) and unrepresented sine
change (right). Top: estimated $\xi^2$ coefficient, with its true value
shown only in the represented case. Bottom: nonlinear prediction RMSE
on the common grid. The plant changes at 50~s; dotted lines mark later
reference changes. The known-model grid error is zero and omitted.}
\label{fig:study4-main}
\end{figure}

%% file: outside_model_class.tex
\subsubsection{Change outside the retained model class}
\label{sec:num-outside}

We use the same linear-start plant and compared models, but set
\begin{equation}
 h_k=0,\qquad
 \rho_k=\rho_\star\mathbf 1_{\{t_k\geq50\}},\qquad
 \rho_\star=h_\star
 \sqrt{\frac{1/5}{1/2-\sin(4)/8}}.
 \label{eq:num-outside-change}
\end{equation}
This choice gives $h_\star x^2$ and $\rho_\star\sin(2x)$ the same
root-mean-square magnitude for uniformly distributed $x\in[-1,1]$.
The scaling matches the total root-mean-square magnitude of the added
terms, not their approximation errors within the retained model classes
or the overall difficulty of the two tasks.

The new sine term lies outside every retained polynomial dictionary on a
nontrivial state interval. Initialization, reference, forgetting,
change time, and evaluation queries remain unchanged.
No exact post-change parameter vector is assigned to these polynomial
models. Their forward approximation and closed-loop response are evaluated
rather than coefficient recovery.
The three cases distinguish unnecessary capacity, activated capacity,
and a change outside the retained span.

\begin{table}[!htb]
\centering
\small
\setlength{\tabcolsep}{3pt}
\renewcommand{\arraystretch}{1.13}
\caption{Study 4 deterministic comparison.
Tracking windows are the whole run $[0,120)$, event $[50,60)$,
and late interval $[90,100)$~s. The grid score uses the late interval.
Recovery delays are measured from 50~s; censored denotes no qualifying
dwell in $[50,60)$~s, and a dash denotes not applicable.}
\label{tab:study4-results}
\begin{tabular}{@{}lrrrrc@{}}
\toprule
& \multicolumn{3}{c}{Tracking RMSE} &
\shortstack{Late grid\\RMSE} & \shortstack{Recovery\\delay (s)} \\
\cmidrule(lr){2-4}
Model & Whole & Event & Late & & \\
\midrule
\multicolumn{6}{@{}l}{\textit{No change}} \\
Affine & 0.04958 & 0.05088 & 0.02111 & $1.040\times10^{-5}$ & --- \\
Augmented affine & 0.04958 & 0.05088 & 0.02111 & $2.174\times10^{-5}$ & --- \\
Complete quadratic & 0.04958 & 0.05087 & 0.02111 & $5.724\times10^{-5}$ & --- \\
Known-model reference & 0.04958 & 0.05088 & 0.02111 & $0$ & --- \\
\addlinespace[4pt]
\multicolumn{6}{@{}l}{\textit{Represented nonlinear change}} \\
Affine & 0.05398 & 0.08497 & 0.02311 & $4.275\times10^{-2}$ & Censored \\
Augmented affine & 0.05324 & 0.07373 & 0.02617 & $2.464\times10^{-2}$ & 6.8 \\
Complete quadratic & 0.05394 & 0.08774 & 0.03192 & $2.588\times10^{-2}$ & 5.8 \\
Known-model reference & 0.04923 & 0.04704 & 0.02038 & $0$ & 0.2 \\
\addlinespace[4pt]
\multicolumn{6}{@{}l}{\textit{Unrepresented nonlinear change}} \\
Affine & 0.05478 & 0.07702 & 0.03223 & $2.900\times10^{-2}$ & Censored \\
Augmented affine & 0.05557 & 0.06865 & 0.03511 & $2.906\times10^{-2}$ & 7.4 \\
Complete quadratic & 0.05581 & 0.07836 & 0.03870 & $6.806\times10^{-2}$ & 6.3 \\
Known-model reference & 0.04910 & 0.05143 & 0.02093 & $0$ & 0.2 \\
\bottomrule
\end{tabular}

\smallskip
\parbox{\linewidth}{\footnotesize
Each late grid score pools 441 fixed points at ten one-second
snapshots. The no-change control uses the same scoring windows.}
\end{table}

\paragraph{Unrepresented-change results}
Following the sine addition, the augmented affine and complete quadratic
models again satisfy the recovery criterion, after 7.4 and 6.3~s.
The affine model remains right-censored. Nevertheless, none of these
polynomial models contains the changed plant exactly.
The fitted $\xi^2$ coefficients also become nonzero in this case
(Figure~\ref{fig:study4-main}, top right), although no exact post-change
polynomial coefficient vector is available.
Activation of a retained coefficient is therefore not, by itself,
evidence that the new plant structure has been identified.

The late-window common-grid prediction RMSE is 0.02900 for the affine
model, 0.02906 for the augmented affine model, and 0.06806 for the
complete quadratic model. The affine model also has the lowest whole-run
and late-window tracking errors among the adaptive models in this case.
Thus, the earlier recovery of the enriched models does not extend to
every prediction or tracking measure.

\Needspace{3\baselineskip}
The grid comparison uses the same evaluation points, but the estimates
are obtained from different closed-loop trajectories.
Recent regressor windows are numerically rank deficient during portions
of every adaptive run at the prescribed relative tolerance.
This finite-window observation does not establish singularity of the
accumulated RLS information.
Overall, the pilot distinguishes coefficient activation, prediction
over a specified region, and closed-loop recovery. It does not establish
a general advantage from increasing dictionary size.

%% file: physical_identification.tex
\subsection{Physical Identification and Reconstruction}
\label{sec:num-physical}

This subsection addresses the physical-reconstruction question in
Table~\ref{tab:num-study-map}.
Consider Plant C,
\begin{equation}
 J\dot x=u-d_{\mathrm C}x^3,\qquad
 J=1,\qquad d_{\mathrm C}=0.6.
 \label{eq:num-plant-c}
\end{equation}
Let $x$ be angular velocity in $\mathrm{rad\,s^{-1}}$ and $u$ be torque
in $\mathrm{N\,m}$. The units of $J$ and $d_{\mathrm C}$ are
$\mathrm{kg\,m^2}$ and
$\mathrm{N\,m}/(\mathrm{rad\,s^{-1}})^3$, respectively.
The parameters are constant and the input is held over each sampling interval.
Only endpoint state measurements are supplied to the controller.

The exact interval relation is
\begin{equation}
 T_su_{k-1}=J(x_k-x_{k-1})
 +d_{\mathrm C}\int_{t_{k-1}}^{t_k}x(t)^3\,\mathrm dt.
 \label{eq:num-integral-identity}
\end{equation}
For the integrated physical model ($\mathsf I$), we use
\begin{equation}
 b_k=T_su_{k-1},\qquad
 \Phi_k=\left[y_k-y_{k-1},\,
             \frac{T_s}{2}(y_{k-1}^3+y_k^3)\right],\qquad
 \theta=(J,d_{\mathrm C})^\top.
 \label{eq:num-integral-regression}
\end{equation}
The identification relation has two unknown coefficients.
Its endpoint implementation has quadrature error; measurement noise would
also enter its regressor. Exactness of~\eqref{eq:num-integral-identity}
does not make~\eqref{eq:num-integral-regression} exact.

For reconstruction, we map the raw estimate componentwise into
$[0.2,3]\times[0,2]$, in the stated parameter units.
The raw estimate and covariance remain the RLS state.
Only the mapped estimates are used in
\begin{equation}
 \dot X=\frac{u-\widehat d_{\mathrm C,k}X^3}{\widehat J_k},\qquad
 X(0)=\xi,\qquad
 \widehat F_{\mathsf I,k}(\xi,u)=X(T_s).
 \label{eq:num-physical-forward}
\end{equation}
Every mapping activation is recorded.
The nonlinear forward map is reconstructed from the identified physical
coefficients; it is not required to be linear in them.

We compare the integrated physical, affine, and complete cubic models, the known-model
reference, and the Euler-informed direct model ($\mathsf D$)
\begin{equation}
 \widehat F_{\mathsf D,k}(\xi,u)
 =\xi+\widehat\alpha_{1,k}u+\widehat\alpha_{2,k}\xi^3.
 \label{eq:num-euler-informed}
\end{equation}
For the Euler-informed direct model, $g_{K,\mathsf D}=\xi$ and
$\Psi_{\mathsf D}=[u,\xi^3]$, with two unknown coefficients.
This reduced direct model retains an Euler-informed form but fits its coefficients to
sampled data. The values $T_s/J$ and $-T_sd_{\mathrm C}/J$ are not exact
sampled-fit targets.

The affine and complete cubic models use the regressors in
\eqref{eq:num-b-models}. Only the integrated physical model has a physical parameter-error
target. All models are compared through forward-output errors in common
units, not through unlike inverse and direct regression residuals.

We use $\lambda_k=1$ and
$r_{\mathrm C}(t)=0.7\sin(2\pi\,0.05t)$.
Before fitting, the integrated physical model is initialized with
$(J,d_{\mathrm C})=(0.8,0.3)$.
The Euler-informed direct and complete cubic models are initialized with the common map
$\xi+(T_s/0.8)u-(0.3T_s/0.8)\xi^3$, and the affine model with
$\xi+(T_s/0.8)u$.

We use fourth-order Runge--Kutta propagation with 20 substeps for the plant
and five substeps for the physical predictors of the integrated physical model and
the known-model reference. Predictor derivatives are computed consistently
with these stages. Both integrations require checks by refinement and by
the independent zero-input solution
\begin{equation}
 F_{\mathrm C}(x,0;J,d_{\mathrm C})
 =\frac{x}{\sqrt{1+2d_{\mathrm C}T_sx^2/J}}.
 \label{eq:num-zero-input-flow}
\end{equation}

A separate noise-free audit compares endpoint quadrature with a refined
interval integral at $T_s\in\{0.1,0.05,0.025\}\ \mathrm{s}$ on a fixed
$21\times21$ grid over $x\in[-1,1]$ and $u\in[-0.6,0.6]$.
The refined integral is evaluator-only information.
This offline quadrature audit is not a closed-loop sampling-rate comparison.

\paragraph{Results}
All five deterministic runs completed without rejected identification
updates or control fallbacks. The physical-parameter mapping was inactive
throughout initialization and control. The plant-output and hard input
bounds were respected. The input-increment bound was active at three
samples in each run, without a violation.
Table~\ref{tab:study5-results} separates held-out prediction after the
common initialization record from tracking during subsequent online
adaptation.

\begin{table}[!htbp]
\centering
\small
\setlength{\tabcolsep}{3pt}
\renewcommand{\arraystretch}{1.14}
\caption{Study 5 deterministic comparison: prediction after 200 initialization
transitions and closed-loop performance.}
\label{tab:study5-results}
\begin{tabular}{@{}lr r rr r@{}}
\toprule
 & & \shortstack{Held-out\\prediction} &
 \multicolumn{2}{c}{Tracking RMSE} & \shortstack{Mean\\$|\Delta u|$} \\
\cmidrule(lr){4-5}
Model & $n_\theta$ & RMSE & Whole & Post-transient & \\
\midrule
Integrated physical & 2 & $1.374\times10^{-5}$ & $2.254\times10^{-3}$ & $6.765\times10^{-5}$ & 5.031 \\
Euler-informed direct & 2 & $4.517\times10^{-4}$ & $2.312\times10^{-3}$ & $5.445\times10^{-4}$ & 5.039 \\
Affine & 3 & $6.288\times10^{-3}$ & $1.145\times10^{-2}$ & $1.098\times10^{-2}$ & 4.952 \\
Complete cubic & 10 & $1.468\times10^{-4}$ & $2.264\times10^{-3}$ & $1.842\times10^{-4}$ & 5.037 \\
Known-model reference & 0 & $2.420\times10^{-11}$ & $2.253\times10^{-3}$ & $5.788\times10^{-5}$ & 5.031 \\
\bottomrule
\end{tabular}

\smallskip
\parbox{\linewidth}{\footnotesize
Prediction uses 600 independent evaluation transitions. Tracking windows
are $[0,80)$ and $[20,80)$~s. All RMSEs are in rad\,s$^{-1}$.
Mean input increments use $[20,80)$~s and are expressed in
$10^{-3}$~N\,m per sample.}
\end{table}

After 200 initialization transitions, the integrated physical model has
a held-out prediction RMSE approximately 97\% below that of the
Euler-informed direct model, although both estimate two coefficients.
The complete cubic model improves on the Euler-informed direct approximation,
but uses ten coefficients and has approximately 11 times the prediction
error of the integrated physical model. These comparisons concern the specified
identification and reconstruction procedures, including their scaling
and priors, rather than parameter count alone. The integrated physical model is not
uniformly better during initial fitting: at 25 and 50 transitions, its
prediction error exceeds that of the Euler-informed direct model.

After initialization, the physical estimates are
$\widehat J=0.999995$ and $\widehat d_{\mathrm C}=0.599346$ in the
stated units. At the last stored estimate, $t=79.9$~s, the values are
1.000039 and 0.599809, respectively. Both remain close to, but do not
equal, the true parameters $J=1$ and $d_{\mathrm C}=0.6$.
Their finite-record accuracy does not establish asymptotic convergence
or remove the quadrature error in the identification relation.

The endpoint-quadrature audit separates that error from numerical
propagation. The RMS interval-equation residual divided by $T_s$ is
$1.180\times10^{-3}$, $3.284\times10^{-4}$, and
$8.686\times10^{-5}$~N\,m at sampling intervals of 0.1, 0.05, and
0.025~s, respectively (Figure~\ref{fig:study5-main}(a)).
The corresponding refined-integral RMS residuals are below
$3.2\times10^{-13}$~N\,m. On the same audit grid, the largest state
changes under plant and predictor substep refinement are
$2.54\times10^{-12}$ and $6.97\times10^{-10}$~rad\,s$^{-1}$,
respectively. Thus, the nonzero endpoint residual is not explained by
insufficient resolution of these forward integrations. The audit varies
sampling only for this numerical check, not for the closed-loop comparison.

\begin{figure}[!htbp]
\centering
\includegraphics[width=\linewidth]{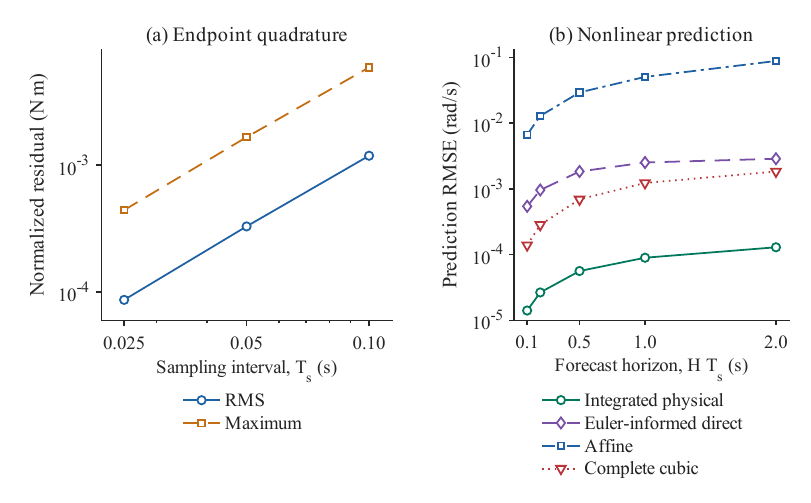}
\caption{Study 5: (a) RMS and maximum endpoint-quadrature residual,
divided by $T_s$, over 441 queries. (b) Nonlinear forecast RMSE after
200 initialization transitions, using 29 common initial states and
rate-limited inputs. Refined-integral and known-model numerical errors
are omitted from (a) and (b), respectively.}
\label{fig:study5-main}
\end{figure}

The integrated physical predictor also has the lowest nonlinear forecast error
among the adaptive models at all five tested horizons, for both held and
rate-limited inputs. Figure~\ref{fig:study5-main}(b) shows the latter case
using the same 29 initial states and input sequences for every model.
Estimates are fixed at their 200-transition values during these forecasts.
The known-model reference has a small, nonzero numerical prediction error
because its integrator differs from the refined reference propagation;
its rate-limited forecast RMSE remains below
$1.22\times10^{-10}$~rad\,s$^{-1}$.

Over the prescribed post-transient interval, the tracking RMSE of the
integrated physical model is about one eighth that of the Euler-informed direct model
and is close to the known-model reference. The nonlinear models
have similar input activity, while their whole-run tracking scores are
less separated than their post-transient scores. These results support
using the physical identification relation to reconstruct an accurate
nonlinear predictor; they do not make the numerical implementation exact
or establish a general ordering of the model classes.

The integrated physical and Euler-informed direct regressions retain full numerical rank
in all 750 recorded 50-transition windows. The complete cubic
regressor is numerically rank deficient in 744 of those windows at the
specified relative threshold of $10^{-10}$, despite completing the run.
This local information diagnostic does not establish singularity of the
accumulated RLS information or, by itself, explain the performance differences.

%% file: diagnostics.tex
\subsection{Prediction and Control Diagnostics}
\label{sec:num-diagnostics}

This subsection addresses the last question in
Table~\ref{tab:num-study-map} and reuses estimates from the preceding
studies. It does not introduce another controller.

\paragraph{Common-query forecasts}
We fix an estimator snapshot, a common initial state $\chi$, and a common
admissible input sequence $\bar u_0,\ldots,\bar u_{H-1}$.
Each model's affine approximation is recomputed at $(\chi,\bar u_0)$.
Let $z_0^\star=z_0^{\mathrm n}=z_0^{\mathrm a}=\chi$ and define
\begin{align}
 z_{i+1}^\star
 &=F_k^\star(z_i^\star,\bar u_i),
 \nonumber\\
 z_{i+1}^{\mathrm n}
 &=\widehat F_{M,k}(z_i^{\mathrm n},\bar u_i),
 \nonumber\\
 z_{i+1}^{\mathrm a}
 &=c_{M,k}+A_{M,k}z_i^{\mathrm a}+B_{M,k}\bar u_i.
 \label{eq:num-common-rollouts}
\end{align}
Here, $F_k^\star$ is the noise-free true map with its current plant regime
held fixed. For Plant C, the refined plant integrator evaluates this map.
The estimates and affine coefficients remain fixed during each forecast.

We evaluate forecasts for $H\in\{1,2,5,10,20\}$ using both held-input and rate-limited
input sequences from the independent evaluation records.
The anchors are taken at indices $20,40,\ldots,580$ of each 600-transition record.
Every model receives the same anchor and input sequence.

Define the terminal model, freezing, and total errors by
\begin{equation}
 e_H^{\mathrm m}=z_H^{\mathrm n}-z_H^\star,\qquad
 e_H^{\mathrm f}=z_H^{\mathrm a}-z_H^{\mathrm n},\qquad
 e_H^{\mathrm t}=z_H^{\mathrm a}-z_H^\star.
 \label{eq:num-decomposed-errors}
\end{equation}
These errors satisfy
\begin{equation}
 e_H^{\mathrm t}=e_H^{\mathrm m}+e_H^{\mathrm f},\qquad
 (e_H^{\mathrm t})^2
 =(e_H^{\mathrm m})^2+(e_H^{\mathrm f})^2
 +2e_H^{\mathrm m}e_H^{\mathrm f}.
 \label{eq:num-decomposition-identity}
\end{equation}
RMS errors use the same query set, and the mean cross term is retained.
At $H=1$, the freezing error is zero to numerical tolerance.
For a complete affine learned map, it is zero at every horizon.
These identities provide checks independent of identification accuracy.

The primary diagnostic uses clean initial states and excludes future plant
changes.

\paragraph{Forecast accuracy}
The diagnostics reuse the fitted models and saved records from the preceding
studies; no additional closed-loop runs are performed.
Table~\ref{tab:study6-results} compares forecasts after 200 noise-free identification
transitions at the one-second horizon, using 29 common initial states and
rate-limited input sequences for each plant.
The estimates and affine coefficients remain fixed during each forecast.
These results are derived from existing study records, not from an independent confirmation study.

The common-query forecasts distinguish nonlinear reconstruction from
the affine model used by the QP.
For Plants A and B, retaining the correct nonlinear structure substantially
reduces nonlinear-model error, but this advantage is reversed by affine
freezing on the tested rate-limited sequences.
For Plant A, the ten-step nonlinear prediction RMSE is
$6.84\times10^{-5}$ for the shared model and $4.00\times10^{-2}$ for the
affine model. After affine freezing, the error for the shared model
increases to $6.53\times10^{-2}$, whereas the affine model is unchanged.
The known-model reference has a similar frozen-prediction error.

\begin{table}[!htbp]
\centering
\small
\setlength{\tabcolsep}{4pt}
\renewcommand{\arraystretch}{1.12}
\caption{Study 6: forecast RMSE at $HT_s=1$~s after 200 identification
transitions, under rate-limited inputs.}
\label{tab:study6-results}
\begin{tabular}{@{}llrrr@{}}
\toprule
Plant & Predictor & \shortstack{Nonlinear\\model} & Freezing & Total \\
\midrule
A & Affine              & $4.003\times10^{-2}$ & $<10^{-15}$           & $4.003\times10^{-2}$ \\
A & Shared              & $6.843\times10^{-5}$ & $6.527\times10^{-2}$ & $6.528\times10^{-2}$ \\
\addlinespace
B & Affine              & $3.923\times10^{-2}$ & $<10^{-15}$           & $3.923\times10^{-2}$ \\
B & Exact sine          & $5.607\times10^{-5}$ & $1.293\times10^{-1}$ & $1.292\times10^{-1}$ \\
\addlinespace
C & Affine              & $5.065\times10^{-2}$ & $<10^{-15}$           & $5.065\times10^{-2}$ \\
C & Integrated physical & $8.963\times10^{-5}$ & $1.484\times10^{-2}$ & $1.489\times10^{-2}$ \\
\bottomrule
\end{tabular}

\smallskip
\parbox{\linewidth}{\footnotesize
Each row uses the same 29 queries within its plant; evaluation input bounds
are $L=0.5$, $1.1$, and $0.6$ for A, B, and C, respectively.
Errors are dimensionless for A and B and in rad\,s$^{-1}$ for C.
They are terminal forecast errors, not closed-loop tracking scores.
The three columns are not additive.}
\end{table}

The exact-sine Plant B model likewise has a larger total frozen-prediction
error than its affine baseline.
For Plant C, the integrated physical model retains a lower total error than
the affine baseline, although freezing substantially reduces its nonlinear
prediction advantage.
Figure~\ref{fig:study6-main} shows the horizon dependence for the exact-sine
and integrated physical models.
The known-model references have nearly the same freezing errors as these
accurately identified models: their total errors at one second are
0.1293 for Plant B and 0.01485~rad\,s$^{-1}$ for Plant C.
Thus, more accurate coefficient estimates alone would not remove the
observed affine approximation error.

\begin{figure}[!htbp]
\centering
\includegraphics[width=\linewidth]{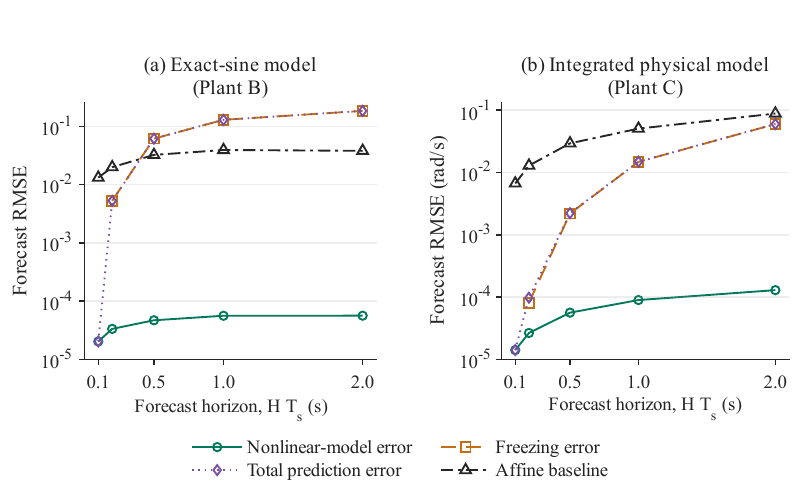}
\caption{Study 6: nonlinear-model, freezing, and total forecast errors.
(a) Exact-sine model on Plant B, $L=1.1$.
(b) Integrated physical model on Plant C.
Each panel uses 200-transition estimates and 29 common initial states
under rate-limited inputs; the affine baseline uses the same queries.
First-step freezing errors are zero to numerical precision and omitted
from the logarithmic axes.}
\label{fig:study6-main}
\end{figure}

The held-input comparisons separate this effect from identification error.
For the shared Plant A model and the exact-sine Plant B model, the freezing
errors remain at numerical roundoff when the input is held at its
linearization value. At fixed input, both reconstructed maps are affine
in the state, so the frozen predictor reproduces their nonlinear rollout.
This property does not hold for the nonlinear state dependence of Plant C.
The integrated physical model for Plant C has a held-input freezing RMSE of
$1.866\times10^{-2}$~rad\,s$^{-1}$ at one second, compared with a
nonlinear-model RMSE of $1.002\times10^{-4}$~rad\,s$^{-1}$.

The signed cross term in the squared-error decomposition is retained.
Model and freezing errors can partly cancel; their mean-square errors
cannot simply be added. The first-step freezing errors are at numerical
roundoff, and complete affine models reproduce their nonlinear forecasts
at every tested horizon. All archived forecast paths are finite.
These diagnostics under prescribed input sequences isolate prediction effects.
This fixed-model forecast comparison is separate from the receding-horizon
closed-loop comparison and does not establish a closed-loop performance
ordering. The receding-horizon controllers receive new measurements and
reconstruct their predictors at subsequent control instants.

\paragraph{Initialization and unforeseen changes}
Separate audits initialize learned forecasts from $y_k$ or cross an actual
plant change. Their total errors include initialization or future-change
effects in addition to the terms isolated above. Neither effect is
attributed to dictionary approximation or affine freezing.

The measurement-initialization audit uses only the saved Study~1
queries. Even the known-model reference has a nonzero error when prediction
starts from the measured rather than the true state.
For held-input forecasts, the median initialization-contribution RMSE over
the 40 archived noisy trials is $6.331\times10^{-3}$ at one step and
$2.650\times10^{-4}$ at ten steps.
The nonlinear-model and freezing contributions are at roundoff in this
particular known-model comparison, so the remaining error is attributable
to initialization. The observed decay applies to these trajectories, not
to arbitrary nonlinear plants.

The separate future-change audit uses the nine noise-free Study~3 abrupt-change cases,
forecast origins 49.5 and 49.9~s, and $H=20$, with both held and recorded
applied-input sequences. Process noise is suppressed during the forecast.
Only the evaluator uses the future gain schedule.
For the known-model held-input forecasts, the absolute future-change
contributions are 0.2514 and 0.2522 at the two origins, while model and
freezing errors remain at roundoff.
Knowledge of the current plant therefore does not eliminate the error due
to an unannounced future change. Recorded inputs are replayed as diagnostic
queries; they are not assumed known at the forecast origin. These
within-run queries also do not constitute a common-input comparison across
controllers.

\paragraph{Along-trajectory forecasts}
We also report stored forecasts originating on the six noise-free stationary
Plant A control trajectories. For each controller, the origins are
\(t_k=2,4,\ldots,58\) s. At each origin, the measured state equals the true
state, and the current online parameter estimate is held fixed during the
forecast. The affine predictor is formed at that origin and is also held
fixed. Nonlinear and affine forecasts use the subsequently recorded applied
inputs, with terminal errors evaluated at \(H\in\{1,2,5,10,20\}\).
The reference is the noise-free Plant A propagation from the same origin
under the same recorded inputs.

Table~\ref{tab:ejcArchiveTrajectory} reports RMS terminal errors over the
29 origins for each controller. At \(H=10\), the shared model has
nonlinear-model, freezing, and total RMSEs of
\(4.433\times10^{-5}\), \(1.537\times10^{-2}\), and
\(1.536\times10^{-2}\), respectively; the affine model has total RMSE
\(4.639\times10^{-2}\). These are retrospective within-run diagnostics,
not the forecasts of the input plans available to the QP at each origin.
The later applied inputs are replayed only for evaluation. The forecast times
are common across controllers, but the states, online
parameter estimates, and replayed input sequences are trajectory-specific,
unlike the common-query comparisons above. The results cover the stated
origins, not every sample of the 120-s runs, and do not establish a
controller ordering under common state--input queries. No new closed-loop
runs are used.

\begin{table}[tbp]
\centering
\caption{Within-run terminal total forecast RMSEs for noise-free stationary
Plant A. Each model uses 29 origins on its own closed-loop trajectory and
the subsequently recorded applied inputs. Errors are dimensionless.}
\label{tab:ejcArchiveTrajectory}
\begingroup
\fontsize{9}{11}\selectfont
\setlength{\tabcolsep}{2.5pt}
\begin{tabular}{lrrrrr}
\hline
Model & \(H=1\) & \(H=2\) & \(H=5\) & \(H=10\) & \(H=20\)\\
\hline
Affine & $1.455\times10^{-2}$ & $2.552\times10^{-2}$ & $4.105\times10^{-2}$ & $4.639\times10^{-2}$ & $5.276\times10^{-2}$ \\
Shared & $1.700\times10^{-5}$ & $2.107\times10^{-3}$ & $1.070\times10^{-2}$ & $1.536\times10^{-2}$ & $3.590\times10^{-2}$ \\
Incorrectly shared & $2.756\times10^{-2}$ & $4.854\times10^{-2}$ & $8.584\times10^{-2}$ & $1.021\times10^{-1}$ & $1.439\times10^{-1}$ \\
Relaxed & $2.369\times10^{-5}$ & $2.102\times10^{-3}$ & $1.070\times10^{-2}$ & $1.536\times10^{-2}$ & $3.589\times10^{-2}$ \\
Complete quadratic & $4.928\times10^{-5}$ & $2.073\times10^{-3}$ & $1.059\times10^{-2}$ & $1.521\times10^{-2}$ & $3.573\times10^{-2}$ \\
Known-model & $<10^{-14}$ & $2.116\times10^{-3}$ & $1.070\times10^{-2}$ & $1.537\times10^{-2}$ & $3.590\times10^{-2}$ \\
\hline
\end{tabular}
\endgroup
\end{table}

\begingroup
\emergencystretch=1em
\Needspace{8\baselineskip}
\paragraph{Constraint audit}
The audit uses four dedicated Plant B constraint cases generated with
Study~2, distinct from its eighteen amplitude runs. These four records
are reused here without additional closed-loop simulations. These runs compare the affine, odd cubic, and
exact-sine models and the known-model reference under the conditions
\begin{equation}
 x_0=1.4,\qquad u_0=0,\qquad r(t)=0.6,\qquad
 x_{\max}=0.85,\qquad v_k=w_k=0.
 \label{eq:num-constraint-case}
\end{equation}
The committed input gives $x_1=0.03+0.65(1.4)=0.94$, independently of
the newly computed input.
Thus the known-model QP requires at least $0.09$ upper-bound slack at
its first predicted output. This required slack is an algebraic consequence
of timing, not a simulated result.

First-step predicted slack and actual next-sample violation are reported
separately. The initial violation is recorded separately from future
violations. The initial output violation is 0.55, and the first future
output violates the bound by 0.09.
The known-model first-step upper slack is 0.09000, whereas the
affine model predicts an upper slack of 0.08544.
Thus, predicted slack and actual violation are distinct even at the first
step. Each of the four runs has only one future output sample outside the
soft bound. These observations confirm the timing benchmark without
turning the soft prediction constraints into a plant-output guarantee.\par
\endgroup

%% file: discussion_conclusions.tex
\section{Discussion and Design Implications}
\label{sec:discussion-design}

The studies distinguish three design decisions: which functions and coefficient
relations to retain, how the unknown coefficients are identified, and how the
reconstructed model is used for prediction in the controller. Improving one
stage does not necessarily improve the next. Table~\ref{tab:study-design-lessons}
summarizes the observations and their implications under the tested conditions.
The resulting guidance is to use structural information to allocate model
complexity, rather than to increase dictionary size indiscriminately.

\begin{table}[!ht]
\centering
\small
\setlength{\tabcolsep}{4pt}
\renewcommand{\arraystretch}{1.08}
\caption{Observations and design implications of the six studies.}
\label{tab:study-design-lessons}
\begin{tabularx}{\linewidth}{@{}c>{\raggedright\arraybackslash}X>{\raggedright\arraybackslash}X@{}}
\toprule
Study & Observation & Design implication \\
\midrule
1 & Correct sharing improved prediction at unchanged coefficient count.
  & Retain credible coefficient relations before adding independent unknowns. \\
\addlinespace
2 & Restricted polynomials were useful without exact plant representation.
  & Assess enrichment against the available data and intended operating region. \\
\addlinespace
3 & Forgetting improved readaptation without accurate recovery of individual coefficients.
  & Evaluate tracking, forgetting activity, and parameter error separately. \\
\addlinespace
4 & Nonlinear coefficients activated after represented and unrepresented changes.
  & Do not interpret coefficient activation as structure discovery. \\
\addlinespace
5 & Physical reconstruction outperformed the two-coefficient Euler-informed approximation after initialization.
  & Distinguish the identification relation from the forward predictor. \\
\addlinespace
6 & Affine freezing could reverse a nonlinear forecast advantage.
  & Check the horizon approximation before further enriching the dictionary. \\
\bottomrule
\end{tabularx}
\end{table}

\Needspace{8\baselineskip}
\paragraph{Retain structure relevant to the task}
When credible functional relations are available, they can be retained while
only their independent unknown coefficients are identified. Study~1 shows that
this can improve prediction without increasing the number of unknowns.
When only partial information is supplied, restrictions such as oddness or
known coefficient relations can instead define an approximate model.
The polynomial models in Study~2 demonstrate that exact plant representation
is not necessary for useful prediction and control. Sufficient local smoothness supports
polynomial approximation, but it does not by itself select a useful degree
or specify which coefficient relations should be imposed.

With little reliable functional information, an affine model remains a useful
baseline for assessing nonlinear enrichment. In the unchanged
linear case of Study~4, additional terms gave essentially no tracking benefit.
The design question is therefore whether a proposed nonlinear contribution
improves the behavior relevant to the task, not simply whether it enlarges the
model class. The operating region and the available identification record
must be considered together. These comparisons do not determine a universal
degree or a neighborhood in which an affine model suffices.

\paragraph{Separate representation from identification}
An exact representation establishes that a suitable coefficient vector exists;
it does not ensure that the available data identify that vector accurately.
The retained functions must be distinguishable along the measured trajectories.
Visiting a region where nonlinear terms are appreciable is not, by itself,
sufficient. The compensating estimates in Study~3 and the recent-data rank
diagnostics are consistent with this distinction, although they do not isolate
insufficient excitation as the sole cause of the observed errors.

Forgetting reduces the influence of past data, but neither generates excitation
nor introduces missing functions. Its effect also depends on the change
schedule and performance measure, as the step and drift comparisons show.
Similarly, the retained nonlinear terms in Study~4 aided recovery after the
represented change, but also activated when approximating an unrepresented
change. Tracking recovery, coefficient recovery, and identification of physical
structure must therefore be assessed separately. A small tracking error does
not establish accurate coefficients, and coefficient activation alone does
not identify a new physical term.

\paragraph{Separate identification from reconstruction}
Study~5 shows that retaining physical information need not require a low-order
sampled-map approximation. A computable integrated relation can identify
physical parameters, which are then used to propagate the nonlinear dynamics.
Only the identification relation must be linear in the unknown parameters.
The integrated physical and Euler-informed direct models both estimate two coefficients, but
use the supplied information differently; their comparison shows why
parameter count alone does not characterize a prediction model.
The quadrature and propagation checks remain necessary: an exact
continuous-time parameterization does not make its numerical realization exact.

\paragraph{Check how the controller uses the learned model}
Study~6 evaluates nonlinear-model error and Jacobian-freezing error
separately under common forecast queries. A single affine approximation, frozen over the horizon, can lose
that accuracy as the predicted state and input depart from the linearization
point. The similar freezing errors of the known-model references show that
more accurate coefficient estimates alone would not remove this error.
This finding concerns the prediction approximation, not a failure of the QP
solver. These forecasts under prescribed input sequences do not establish
a closed-loop performance ordering, since the implemented controllers receive new
measurements and reconstruct their predictors at each control instant.

On the selected noise-free Plant~A trajectories, the shared model has
lower total forecast RMSE than the affine model at every reported
horizon in Table~\ref{tab:ejcArchiveTrajectory}. The shared model also has a lower whole-run
tracking median in the paired noisy comparison in Table~\ref{tab:num-structure-results}.

When nonlinear forecasts are already accurate but frozen-affine forecasts are
not, the results motivate examining the horizon approximation rather than
enlarging the dictionary. Related nonlinear predictive formulations instead update
state- and control-dependent models or predicted piecewise-linear model sequences
iteratively over the horizon~\cite{Kamaldar2026IterativeSCDMPC,IslamBernstein2025PiecewiseLinearMPC}.
A targeted future comparison would use the same
identified nonlinear models with frozen-affine prediction, trajectory-dependent
linearization, and nonlinear model predictive control. It would assess whether
preserving more nonlinear prediction accuracy improves realized control enough
to justify the additional computation. That benefit is not established here.

These implications are drawn from the reported scalar, full-state studies
under the specified initialization, noise, reference, and tuning conditions.
They support conditional design guidance, not a universal choice of dictionary
or forgetting rule. For a new plant, the practical sequence is to retain
credible structure, assess whether the data support its identification, and
then check whether the controller preserves the resulting prediction accuracy.
Additional complexity should address the limitation identified at that stage.

\section{Conclusions}
\label{sec:conclusions}

The numerical investigation examines structure-informed model construction
within a common indirect adaptive predictive-control architecture.
Correct functional and coefficient relations improved prediction without
necessarily increasing the number of estimated parameters. Partial
restrictions also supported useful approximations without exact plant
representation. Physical identification provided a distinct route to accurate
nonlinear reconstruction, beyond direct sampled-map fitting.

The studies also distinguish the limits of these improvements. Additional
coefficients did not ensure better control; forgetting and coefficient
activation did not establish recovery of the true plant parameters or
structure. Accurate nonlinear prediction could, in turn, lose its advantage
after affine freezing. The resulting design principle is to evaluate
representation, identification, reconstruction, and horizon prediction
separately, and to introduce additional complexity where the observed
limitation occurs rather than treating dictionary enlargement as a general
solution.

%% file: numerical_protocol.tex
\section{Numerical Protocol}
\label{app:num-protocol}

\paragraph{Measurement and control conventions}
The state, input, and output dimensions are one.
The controller uses the measured state. The true state is available only
to the evaluator. No observer is introduced.

Use the RLS estimator, reconstruction map, and constrained QP defined in Sections~\ref{sec:c4-identification} and~\ref{sec:c4-control}.
At time $t_k$, process the transition generated by $u_{k-1}$ and update the
model. The input $u_k$ is already committed. Initialize the prediction trajectory with
$(\xi_{0|k},u_{0|k})=(y_k,u_k)$ and apply
$u_{k+1}=u_{1|k}^{\star}$.
The model parameters and affine coefficients remain fixed within each
prediction horizon.

The constraints use $|u|\leq u_{\max}$,
$|\Delta u|\leq\Delta u_{\max}$, and softened $|\xi|\leq x_{\max}$.

\paragraph{Model conventions}
Let $M$ denote a model and let $n_{\theta,M}$ be its number of independent
coefficients. A direct model has the form
\begin{equation}
 \widehat F_{M,k}(\xi,u)
 =g_{K,M}(\xi,u)+\Psi_M(\xi,u)\widehat\theta_{M,k}.
 \label{eq:num-direct-model}
\end{equation}
The regressor rows listed in Sections~\ref{sec:num-complexity}--\ref{sec:num-physical}
are complete and ordered.
Unless stated otherwise, set $g_{K,M}=0$, use the identity admissibility
map, and estimate the listed coefficients independently.
Coefficient sharing is imposed within the regressor, not by estimating
repeated copies of an unknown.
All dictionaries are fixed before identification begins.

The known-model reference, denoted $\mathsf K$, uses the current true
parameters, the same measurements, and the same frozen-affine QP.
For Plant C, the known-model reference also uses the same numerical predictor
integrator. This reference receives no future parameter schedule and is
diagnostic, not a guaranteed best controller.

\paragraph{Initialization and common-data evaluation}
For the primary comparisons, generate a separate 200-transition
initialization record from each plant's initial regime. Fit every adaptive model to this record with
$\lambda_k=1$. Carry its fitted coefficients and covariance into the
control run, and restart the plant at the prescribed initial condition.
Do not join the two records through a fictitious transition.
No new control-run transition is available for an update at $k=0$.

For the initialization and evaluation records, draw an input target every
five samples from $L\{-1,-1/2,0,1/2,1\}$ with equal probabilities.
Move the applied input toward this target subject to the input-increment
bound. Use initialization scales $L=0.50$, $1.10$, and $0.60$ for
Plants A, B, and C, respectively.
Independent 600-transition evaluation records use $L=0.50$ for A,
$L\in\{0.25,0.70,1.10\}$ for B, and $L=0.60$ for C.

Evaluate the estimates after 25, 50, 100, and 200 initialization transitions
on the same noise-free evaluation records.
These ideal-state, common-data tests separate model construction
and finite-data identification from the different trajectories generated
by feedback. Closed-loop comparisons do not provide identical subsequent
identification data across controllers.

\paragraph{Scaling and initial uncertainty}
Apply a fixed row scaling to each identification relation.
Divide both sides of a direct regression by one output unit, and both
sides of the physical inverse regression by $T_s$ times one input unit.
Let $D_M$ contain the positive root-mean-square regressor-column magnitudes
on a fixed calibration grid. Define
\begin{equation}
 \overline\Phi_{M,k}=\widetilde\Phi_{M,k}D_M^{-1},\qquad
 \beta_M=D_M\theta_M,\qquad
 P_{\beta,M,\mathrm{init}}
 =\frac{100}{n_{\theta,M}}I_{n_{\theta,M}}.
 \label{eq:num-scaling}
\end{equation}
Here, $\widetilde\Phi_{M,k}$ is the row-scaled regressor.
Apply RLS in the $\beta_M$ coordinates and map back before reconstruction.

For direct models, use a $41\times41$ calibration grid over
$\xi\in[-1.2,1.2]$ and the admissible input interval.
For the integrated physical model, use 41 previous-state values in $[-1.2,1.2]$
and 41 measured increments in $[-0.1,0.1]$, in the stated state units.
These grids define scaling, not additional identification data.
The convention equalizes average prior regression variance on each
calibration grid, not physical parameter uncertainty across different
identification relations. Raw initial models are specified in Sections~\ref{sec:num-complexity}--\ref{sec:num-physical}.

\paragraph{Noise and repetitions}
The primary comparisons use $v_k=w_k=0$, where $w_k$ is the additive
process disturbance in the sampled plants.
Repeat the stationary structural-information comparison with 40 paired
realizations of $v_k\sim\mathcal N(0,0.01^2)$ and $w_k=0$.
For the selected gain-step comparison, use 40 paired realizations with
$v_k\sim\mathcal N(0,0.01^2)$ and
$w_k\sim\mathcal N(0,0.005^2)$.
Use the shared model with the three forgetting rules, its pretrained--frozen
control, and the known-model reference in this second comparison.

The noise sequences are mutually and temporally independent, and are
paired across models by sample index.
Keep initialization inputs and references fixed across the paired trials.
Initialization noise and control-run noise are independent.
The physical-identification study remains deterministic.

\paragraph{Computational environment and random records}
The saved environment records for Studies 1--5 identify MATLAB
26.1.0.3346908 (R2026a), Update 5, and Optimization Toolbox 26.1.
The derived Study 6 diagnostics record the same MATLAB version and
reuse the source records without generating new random streams.
The initialization and evaluation input seeds are 2101 and 2102 for
Study 1, 4101 and 4102 for Study 3, 5101 and 5102 for Study 4, and
6101 and 6102 for Study 5. Study 2 uses initialization seed 3101 and
evaluation seeds 3102, 3103, and 3104 in increasing evaluation-range order.
For noisy trial \(j=1,\ldots,40\), Study 1 uses measurement-noise seeds
\(2200+j\) for initialization and \(2300+j\) during control.
Study 3 uses initialization measurement/process seeds
\(4200+j\) and \(4300+j\), and control measurement/process seeds
\(4400+j\) and \(4500+j\).
The stored bootstrap seeds are 2401 for Study 1 and 4601 for Study 3.
The noise-free cases use the same input records with noise suppressed.
These seeds identify the retained numerical records.
The sensitivity runs used the same MATLAB and Optimization Toolbox versions,
with source revision \texttt{d1f1adb0d789}. They used \texttt{quadprog} with
\texttt{interior-point-convex}, constraint and optimality tolerances of
$10^{-8}$, and a normalized KKT rejection threshold of $10^{-7}$.
The paired-bootstrap calculation used \texttt{mt19937ar}, seed 7401, and
2,000 resamples. Its 95\% intervals are contrast-specific, not simultaneous.
These execution records identify the sensitivity calculation, not the
historical source revision or solver settings of the primary studies.

\paragraph{Measurements}
Let $K=T/T_s$, where $T$ is the run duration, and define
\begin{equation}
 \mathcal W(a,b)=\{k:0\leq k<K,\ a\leq kT_s<b\}.
 \label{eq:num-window}
\end{equation}
For a nonempty window $\mathcal W$, let
$\mathcal W_\Delta=\mathcal W\cap\{1,\ldots,K-1\}$ and define
\begin{align}
 E_{\mathcal W}
 &=\left(\frac{1}{|\mathcal W|}
       \sum_{k\in\mathcal W}(x_k-r_k)^2\right)^{1/2},
 \nonumber\\
 U_{\mathcal W}
 &=\left(\frac{1}{|\mathcal W|}
       \sum_{k\in\mathcal W}u_k^2\right)^{1/2},
 \nonumber\\
 V_{\mathcal W}
 &=\frac{1}{|\mathcal W_\Delta|}
       \sum_{k\in\mathcal W_\Delta}|u_k-u_{k-1}|.
 \label{eq:num-metrics}
\end{align}
The last quantity is used only when $\mathcal W_\Delta$ is nonempty.
Report tracking error in state units, input RMS in input units, and mean
input increment per sample. The saved summaries also contain peak tracking error and total input
variation, which are not tabulated here. Actual bound exceedances are
recorded separately.
Score the prospective prediction $\widehat F_{M,k}(y_k,u_k)$ against
$x_{k+1}$, not against the transition just used for fitting.
Measured-output scores are recorded separately.
Common-query prediction errors are defined in
Section~\ref{sec:num-diagnostics}.

For stationary Plant A, pool the last 10~s of each 20-s reference plateau.
For Plants B and C, use $\mathcal W(20,80)$.
For abrupt and structural changes, use $\mathcal W(30,40)$ before the event,
$\mathcal W(50,60)$ during the event response, and $\mathcal W(90,100)$ for
the late response. Whole-run scores are reported separately.

The Study~1 late-plateau scores and the pre-event scores are retained in
the saved summaries but are not tabulated here. The pooled late-plateau
windows include reference-preview intervals and therefore do not
isolate settled tracking.

For an abrupt change at 50~s, recovery is the first sampled time
$t\geq50$~s for which $|x_k-r_k|\leq0.04$ holds at every sample
in the half-open interval $[t,t+2)$.
At $T_s=0.1$~s, this requires 20 consecutive samples.
The complete 2-s dwell interval must end no later than the reference change
at 60~s; its right endpoint may equal 60~s.
Recovery time is measured from 50~s, so zero is permitted if the condition
holds throughout $[50,52)$.
Non-recovery is right-censored at 60~s, not assigned a recovery time of 10~s.
Termination before that time is reported separately.

Parameter error is evaluated only when an exact reference vector is
specified. Component errors are recorded, together with the scaled error
\begin{equation}
 E_{\theta,M,k}
 =\frac{100}{\sqrt{n_{\theta,M}}}
 \left\|D_{\theta,M}^{-1}
 (\widehat\theta_{M,k}-\theta_{M,k}^{\star})\right\|_2,
 \label{eq:num-parameter-error}
\end{equation}
where $D_{\theta,M}$ contains fixed component scales.
Use one unit for each dimensionless coefficient and one stated physical
unit for each physical coefficient. Raw and mapped physical estimates are recorded separately. The scaled
error in \eqref{eq:num-parameter-error} is not tabulated here.

\paragraph{Covariance and QP conditioning}
The covariance diagnostic is \(\kappa_2(P_k)\) in the normalized
\(\beta=D_M\theta\) coordinates of (A.2). It includes the prior,
initialization record, and subsequent identification data, with the
specified forgetting rule; it is not the recent-regressor Gram matrix.
The QP diagnostic is \(\kappa_2(H_k^{\mathrm{QP}})\), where
\(H_k^{\mathrm{QP}}\) is the condensed Hessian including both input and
output-slack variables. For these scalar \(N=10\) problems, this matrix
has dimension 29: nine future inputs and twenty nonnegative output
slacks. The slack penalty is included, so the diagnostic is not the
condition number of the input block alone. Both condition numbers depend
on the specified coordinates and weights.

Table~\ref{tab:ejcArchiveConditioning} summarizes the stored histories from
503 complete control runs. Study 1 includes six noise-free runs and
40 noisy trials for each of six models; the earlier 12 exploratory runs
are excluded. Study 2 includes the 18 amplitude runs and four constraint
runs; Studies 3--5 include all their reported control cases.
The summaries retain every recorded control-step entry from the selected runs.
All 590,000 QP
condition entries and all 483,500 applicable covariance entries are
finite. The 106,500 known-model covariance entries are not applicable,
not failed measurements. Fixed-model controls in Study 3 retain their
constant prior or fitted covariance and are included as such.
The recorded control steps contain no rejected control computations or
rejected attempted online identification updates. These observations do
not establish general stability, convergence, or a causal relationship
between conditioning and tracking error.

\begin{table}[tbp]
\centering
\caption{Maxima of the stored condition-number histories over the complete
selected control runs. The counts \(n_P\) and \(n_H\) refer to covariance and QP entries.
Covariance counts exclude known-model runs;
the QP counts include them. No applicable nonfinite entry is excluded.
Study 6 introduces no additional control run.}
\label{tab:ejcArchiveConditioning}
\small
\begin{tabular}{crrrrr}
\hline
Study & Runs & \(n_P\) & \(n_H\) &
\(\max\kappa_2(P_k)\) & \(\max\kappa_2(H_k^{\mathrm{QP}})\)\\
\hline
1 & 246 & 246,000 & 295,200 & $8.8774\times10^{2}$ & $7.1685\times10^{4}$ \\
2 & 22 & 12,300 & 14,800 & $3.3382\times10^{4}$ & $1.0162\times10^{5}$ \\
3 & 218 & 211,200 & 261,600 & $8.5401\times10^{2}$ & $1.8568\times10^{5}$ \\
4 & 12 & 10,800 & 14,400 & $1.3052\times10^{3}$ & $7.4769\times10^{4}$ \\
5 & 5 & 3,200 & 4,000 & $1.0125\times10^{3}$ & $2.1738\times10^{5}$ \\
\hline
\end{tabular}
\end{table}

\paragraph{Finite-window numerical rank}
For each recorded full 50-transition window, let $G_k$ be the Gram matrix
of the column-scaled identification regressors, excluding the RLS prior.
Write its singular values as $s_{1,k}\geq\cdots\geq s_{n_\theta,k}\geq0$.
With $\tau=10^{-10}$, define the numerical rank
\begin{equation}
 r_\tau(G_k)=\#\{j:s_{j,k}>\tau s_{1,k}\}.
\end{equation}
The zero matrix has numerical rank zero. Use the reporting convention
\begin{equation}
 \kappa_\tau(G_k)=
 \begin{cases}
 s_{1,k}/s_{n_\theta,k}, & r_\tau(G_k)=n_\theta,\\
 +\infty, & r_\tau(G_k)<n_\theta.
 \end{cases}
\end{equation}
Here, $+\infty$ flags numerical rank deficiency at the stated threshold;
it does not assert exact algebraic singularity. Report the rank range and
the number of deficient windows, rather than interpreting extremely large
finite condition estimates as precise measurements. A finite-only summary,
when given, must identify the number of windows it excludes.

Re-evaluation of the saved matrices gives
Table~\ref{tab:study2-rank}. Each adaptive amplitude run contains 750
recorded full windows. The affine, exact-sine, and odd cubic models have
full numerical rank in every such window at each amplitude. The known-model
reference performs no identification, so its identification-rank diagnostic
is not applicable.
These diagnostics describe recent data. Since the experiment uses no
forgetting, RLS also retains the initialization record and earlier
transitions. Finite-window rank deficiency neither establishes singularity
of the RLS covariance nor demonstrates a control failure. No regularization,
clipping, or change to RLS is introduced by this reporting convention.

\begin{table}[tb]
\centering
\small
\setlength{\tabcolsep}{5pt}
\renewcommand{\arraystretch}{1.12}
\caption{Numerical ranks of the recorded full-window Gram matrices in
Study 2. Deficient-window counts use the relative threshold $10^{-10}$.}
\label{tab:study2-rank}
\begin{tabular}{@{}crrrr@{}}
\toprule
& \multicolumn{2}{c}{Complete cubic ($n_\theta=10$)} &
\multicolumn{2}{c}{Odd quintic ($n_\theta=5$)} \\
\cmidrule(lr){2-3}\cmidrule(lr){4-5}
$A_r$ & Rank range & Deficient/total & Rank range & Deficient/total \\
\midrule
0.2 & 5--8 & 750/750 & 4 & 750/750 \\
0.6 & 6--8 & 750/750 & 5 & 0/750 \\
1.0 & 7--9 & 750/750 & 5 & 0/750 \\
\bottomrule
\end{tabular}
\end{table}

\paragraph{Study 3: complete deterministic comparison}
Table~\ref{tab:study3-deterministic} retains every controller in the noise-free comparison. The event and late windows include the prescribed
reference preview. A dash in the recovery column denotes right-censoring
at 60~s, not a numerical failure. No settling time is assigned to drift.

\begin{table}[tb]
\centering
\small
\setlength{\tabcolsep}{3pt}
\renewcommand{\arraystretch}{1.12}
\caption{Study 3 deterministic comparison. Errors are tracking RMSE on the
indicated time windows. Recovery delay is measured from 50~s.}
\label{tab:study3-deterministic}
\begin{tabularx}{\linewidth}{@{}>{\raggedright\arraybackslash}Xrrcrr@{}}
\toprule
& \multicolumn{3}{c}{Abrupt change} & \multicolumn{2}{c}{Drift} \\
\cmidrule(lr){2-4}\cmidrule(lr){5-6}
Model and update & $[50,60)$ & $[90,100)$ &
\shortstack{Recovery\\(s)} & $[50,90)$ & $[90,100)$ \\
\midrule
Affine: no forgetting & 0.12858 & 0.07393 & --- & 0.07526 & 0.09939 \\
Affine: fixed & 0.14912 & 0.03587 & 5.6 & 0.05471 & 0.03848 \\
Affine: variable-rate & 0.11578 & 0.03977 & --- & 0.06734 & 0.05253 \\
Shared: no forgetting & 0.14753 & 0.07184 & --- & 0.07349 & 0.09833 \\
Shared: fixed & 0.14250 & 0.02605 & 5.8 & 0.05221 & 0.03179 \\
Shared: variable-rate & 0.12444 & 0.02899 & --- & 0.06579 & 0.04501 \\
Shared: pretrained--frozen & 0.16619 & 0.12599 & --- & 0.08071 & 0.12598 \\
Shared: prior--frozen & 0.56963 & 0.38548 & --- & 0.50795 & 0.38550 \\
Known-model reference & 0.05009 & 0.02200 & 0.2 & 0.04788 & 0.02200 \\
\bottomrule
\end{tabularx}
\end{table}

All runs completed without rejected updates or control computations.
Only the prior--frozen shared-model control exceeded the soft output bound;
its maximum excess was approximately 0.5486 in both schedules.
No applied hard input or input-increment violation was recorded.

The 50-transition normalized Gram matrices use the same numerical-rank
threshold as Study 2. In the noise-free abrupt-change runs, the fractions
of full windows classified as rank deficient are 65.8\%, 56.8\%, and
58.7\% for the shared model with no, fixed, and variable-rate forgetting,
respectively. These diagnostics concern recent data and do not replace
the covariance or accumulated-information checks.

\paragraph{Verification and reporting boundaries}
The saved verification records associated with Studies 4, 5, and 6
contain 97, 108, and 118 tests, respectively, with all tests marked
passed in both their before- and after-study lists. The saved
Study 6 diagnostic report contains 27,296 finite forecast paths and
no failure record. Its largest error-decomposition identity residual
is \(2.23\times10^{-16}\) or smaller, its largest squared-identity
residual is \(5.79\times10^{-16}\) or smaller, and its largest
first-step freezing error is \(1.12\times10^{-16}\) or smaller.
The physical quadrature and integration discrepancies are reported
separately in Section 6.4. None of these numerical checks requires
a particular performance ordering.

A rejected RLS update retains the last valid estimate and covariance.
A rejected control computation applies $u_{k+1}=u_k$.
Terminate a run when the true state is nonfinite or $|x_k|>5$ in the
stated state units. Record every rejection, fallback, and termination.
Report finite-prefix scores as incomplete, not as complete-run averages.
Soft output constraints are not plant-safety certificates.

For the repeated-trial comparisons, Table~\ref{tab:ejcArchivePaired}
reports selected median within-pair differences and the stored 95\%
paired-bootstrap intervals from 2,000 resamples. A difference is the
first-listed score minus the second-listed score; its median need not
equal the difference of the two marginal medians. Every listed contrast
has 40 attempted, complete, and finite pairs. The intervals are
contrast-specific summaries, not simultaneous confidence intervals.
The prediction contrasts use the independent evaluation record after
200 noisy fitting transitions. Tracking contrasts retain the stated
whole-run, event, and late windows, including their reference-preview
effects. No new bootstrap samples or closed-loop runs are generated
for this reporting addition.

\begin{table}[tbp]
\centering
\caption{Selected paired RMSE differences from the existing noisy
comparisons. Every row has 40 complete finite pairs and 2,000 stored
bootstrap resamples. Differences are in dimensionless state units.
Positive differences indicate a smaller score for the second-listed
case, not a general controller advantage.}
\label{tab:ejcArchivePaired}
\small
\setlength{\tabcolsep}{2.5pt}
\begin{tabular}{lrrr}
\hline
Contrast & Median & 95\% lower & 95\% upper\\
\hline
\multicolumn{4}{l}{\textit{Study 1 prediction after 200 fitting transitions}}\\
Shared -- affine & $-1.7300\times10^{-2}$ & $-1.7452\times10^{-2}$ & $-1.7212\times10^{-2}$ \\
Shared -- incorrect sharing & $-3.5627\times10^{-2}$ & $-3.5784\times10^{-2}$ & $-3.5533\times10^{-2}$ \\
Shared -- relaxed & $-6.1368\times10^{-5}$ & $-1.2584\times10^{-4}$ & $-3.1880\times10^{-5}$ \\
Relaxed -- complete quadratic & $-2.6795\times10^{-4}$ & $-3.6532\times10^{-4}$ & $-1.4598\times10^{-4}$ \\
\hline
\multicolumn{4}{l}{\textit{Study 1 tracking \([0,120)\) s}}\\
Shared -- affine & $-6.8967\times10^{-3}$ & $-6.9616\times10^{-3}$ & $-6.8545\times10^{-3}$ \\
Shared -- incorrect sharing & $-2.7432\times10^{-2}$ & $-2.7530\times10^{-2}$ & $-2.7369\times10^{-2}$ \\
Shared -- relaxed & $1.9365\times10^{-5}$ & $7.6787\times10^{-6}$ & $3.0619\times10^{-5}$ \\
Relaxed -- complete quadratic & $2.1519\times10^{-5}$ & $-1.1491\times10^{-5}$ & $5.0351\times10^{-5}$ \\
\hline
\multicolumn{4}{l}{\textit{Study 3 tracking \([50,60)\) s}}\\
No forgetting -- fixed & $6.6423\times10^{-3}$ & $5.2513\times10^{-3}$ & $7.7218\times10^{-3}$ \\
No forgetting -- variable-rate & $2.4174\times10^{-2}$ & $2.3894\times10^{-2}$ & $2.4422\times10^{-2}$ \\
\hline
\multicolumn{4}{l}{\textit{Study 3 tracking \([90,100)\) s}}\\
No forgetting -- fixed & $4.8108\times10^{-2}$ & $4.7894\times10^{-2}$ & $4.8358\times10^{-2}$ \\
No forgetting -- variable-rate & $4.4617\times10^{-2}$ & $4.4315\times10^{-2}$ & $4.4824\times10^{-2}$ \\
\hline
\end{tabular}
\end{table}

\paragraph{Sensitivity to tuning and initialization}
We vary three settings of the stationary Plant~A comparison one at a time.
Let $\alpha_P>0$ multiply the normalized initial covariance in (A.2), so that
$P_{\beta,M,\mathrm{init}}=\alpha_P(100/n_{\theta,M})I_{n_{\theta,M}}$,
and let $n_{\mathrm{init}}$ be the number of initialization transitions.
The nominal setting is $(R,\alpha_P,n_{\mathrm{init}})=(0.05,1,200)$.
The six alternatives change only $R$ to $0.025$ or $0.10$, only $\alpha_P$
to $0.1$ or $10$, or only $n_{\mathrm{init}}$ to $50$ or $100$.
All five adaptive models in Table~4 are tested without forgetting.
The known-model reference is evaluated only at the three distinct values
of $R$, since it performs no identification.

The comparisons reuse the stationary study's input records and 40 paired
measurement-noise realizations. Shorter initialization uses the corresponding
prefix of each record and transfers both the fitted coefficients and covariance
to the restarted control run. The raw initial models, scaling grids, reference,
bounds, and control-run noise remain fixed. Prediction RMSE uses the same
600 clean held-out transitions after fitting; tracking RMSE uses the subsequent
control interval $[0,120)$~s. Input increments use samples $1,\ldots,1199$.
The 38 model--configuration combinations give 1,558 runs: 38 noise-free
and 1,520 noisy. The twelve nominal gate cases passed all 204 recorded
comparisons with the saved stationary runs, with zero recorded differences.
All 1,558 runs completed without rejected initialization or online updates,
control fallbacks, or recorded output, input, or input-increment violations.

\begin{table}[tbp]
\centering
\caption{Stationary Plant~A sensitivity: medians over 40 paired noisy trials.
Only the indicated setting differs from the nominal
$(R,\alpha_P,n_{\mathrm{init}})=(0.05,1,200)$.
Panel~(a) reports $10^3$ times the held-out prediction RMSE after initialization;
panel~(b) reports the unscaled tracking RMSE over $[0,120)$~s.
All errors are in dimensionless state units.}
\label{tab:ejc_p06_sensitivity}
\small
\setlength{\tabcolsep}{4pt}
\renewcommand{\arraystretch}{1.12}
\begin{tabular}{@{}lrrrrr@{}}
\hline
Setting & Affine & Shared & \shortstack{Incorrectly\\shared} & Relaxed &
\shortstack{Complete\\quadratic} \\
\hline
\multicolumn{6}{@{}l}{\textit{(a) Held-out prediction RMSE multiplied by $10^3$}} \\
Nominal & 17.86 & 0.5305 & 36.18 & 0.6633 & 0.9505 \\
$R=0.025$ & 17.86 & 0.5305 & 36.18 & 0.6633 & 0.9505 \\
$R=0.10$ & 17.86 & 0.5305 & 36.18 & 0.6633 & 0.9505 \\
$\alpha_P=0.1$ & 17.83 & 0.6240 & 36.15 & 0.7450 & 1.377 \\
$\alpha_P=10$ & 17.86 & 0.5203 & 36.18 & 0.6730 & 0.9492 \\
$n_{\mathrm{init}}=50$ & 23.01 & 1.439 & 48.74 & 2.039 & 2.548 \\
$n_{\mathrm{init}}=100$ & 20.83 & 0.6867 & 42.69 & 1.060 & 1.353 \\
\hline
\multicolumn{6}{@{}l}{\textit{(b) Whole-run tracking RMSE}} \\
Nominal & 0.055965 & 0.049065 & 0.076507 & 0.049053 & 0.049029 \\
$R=0.025$ & 0.055669 & 0.049064 & 0.075576 & 0.049050 & 0.049024 \\
$R=0.10$ & 0.056404 & 0.049101 & 0.077659 & 0.049084 & 0.049052 \\
$\alpha_P=0.1$ & 0.055963 & 0.049069 & 0.076530 & 0.049055 & 0.048971 \\
$\alpha_P=10$ & 0.055965 & 0.049064 & 0.076505 & 0.049053 & 0.049035 \\
$n_{\mathrm{init}}=50$ & 0.056145 & 0.049061 & 0.076176 & 0.049011 & 0.048938 \\
$n_{\mathrm{init}}=100$ & 0.056066 & 0.049065 & 0.076023 & 0.049025 & 0.048969 \\
\hline
\end{tabular}
\end{table}

Table~\ref{tab:ejc_p06_sensitivity} reports noisy-trial medians for the five
adaptive models. At every tested setting, the shared model retains smaller
prediction and tracking medians than the affine and incorrectly shared models,
which estimate the same number of coefficients. Changing $R$ leaves the
post-fit prediction scores unchanged because the initialization fit is
unchanged. It does affect control activity. For the shared model, the median
mean absolute input increment is $0.016964$, $0.014860$, and $0.012560$ at
$R=0.025$, $0.05$, and $0.10$, respectively.

Shorter initialization increases the median held-out prediction error of every
adaptive model, but its tracking effect depends on the model. For the complete
quadratic model, reducing initialization from 200 to 50 transitions raises
prediction RMSE from $9.505\times10^{-4}$ to $2.548\times10^{-3}$, while
tracking RMSE decreases from $0.049029$ to $0.048938$. The median paired
tracking change, defined as the 50-transition score minus the nominal score,
is $-4.272\times10^{-5}$, with a 95\% paired-bootstrap interval
$[-1.433\times10^{-4},-6.361\times10^{-6}]$. This median of paired differences
is distinct from the difference of the two marginal medians.

Prior-covariance effects also depend on the model and score. Reducing
$\alpha_P$ to $0.1$ increases prediction error for the shared, relaxed, and
complete quadratic models; the complete quadratic model nevertheless has
lower tracking error. Increasing $\alpha_P$ to $10$ does not uniformly
improve both scores. These results preserve the stationary study's distinction
between prediction and tracking, but do not show that poorer prediction causes
better control. Online adaptation and controller-specific data remain part of
the tracking comparison. The sensitivity assessment is limited to the tested
one-factor changes on reused records; it does not test interactions among
settings, changing plants, or independent confirmation.

\paragraph{Cross-platform numerical verification}
\label{par:cross-platform-verification}
The reported numerical values are retained as the historical reference.
Two complete reproductions, one on macOS and one on Windows, used the
specified experiments and random records. The
numerical-acceptance protocol was developed through diagnostic runs and
frozen before both successful full reproductions. Agreement with the
reference and independent checks of mathematical identities and solver
validity are separate requirements.

Exact checks cover the specified configurations, seeds, input and noise
records, dimensions, applicability masks, discrete outcomes, and source
relationships. Each applicable finite numerical value $v$ is compared with
its reference $v_{\mathrm{ref}}$ using
\begin{equation}
    |v-v_{\mathrm{ref}}|
    \leq a_v+r_v|v_{\mathrm{ref}}|,
    \label{eq:cross-platform-agreement}
\end{equation}
where $a_v$ is an absolute tolerance in the quantity's declared units and
$r_v$ is dimensionless. The tolerances in
Table~\ref{tab:cross-platform-tolerances} are applied elementwise in the
original coordinates. A quantity required to be finite fails verification if it is nonfinite.
Prescribed not-applicable and infinity conventions are checked separately. The
complete quantity assignments, source checks, and narrowly specified
metadata exceptions are retained in the versioned verification
specification. These tolerances are numerical acceptance requirements,
not certified bounds on simulation error.

\begin{table}[htbp]
\centering
\caption{Cross-platform numerical-agreement tolerances. Absolute values
are in the corresponding declared unit. The specified Study~1 rules take
precedence. $M_{\mathrm{ref}}$ is the reference covariance or Gram matrix
in the normalized coordinates.}
\label{tab:cross-platform-tolerances}
\small
\renewcommand{\arraystretch}{1.12}
\begin{tabular}{@{}>{\raggedright\arraybackslash}p{0.52\linewidth}p{0.29\linewidth}p{0.12\linewidth}@{}}
\hline
Quantity & Absolute & Relative \\
\hline
Study~1 fitted coefficients and covariance, and 200-transition initialization prediction RMSE
& $10^{-10}$ & $10^{-10}$ \\
Specified Study~1 state, output, input, prediction, and online-parameter histories
& $10^{-7}$ & $10^{-7}$ \\
Other trajectories, forecasts, parameters, regression residuals, complete QP
plans and slacks, metrics, bootstrap summaries, and resolved condition estimates
& $10^{-8}$ & $10^{-7}$ \\
Calibration factors
& $10^{-12}$ & $10^{-10}$ \\
Forgetting factors and residual energy, covariance and Gram entries, and
physical-audit residuals and refinement differences
& $10^{-10}$ & $10^{-7}$ \\
Covariance and Gram eigenvalues
& $10^{-10}\max\{1,\|M_{\mathrm{ref}}\|_2\}$ & $10^{-7}$ \\
Normalized QP residuals
& $10^{-9}$ & $10^{-7}$ \\
Analytical-identity residuals
& $10^{-12}$ & $10^{-7}$ \\
\hline
\end{tabular}
\end{table}

Numerical agreement cannot override a failed validity check. An accepted
QP requires a positive solver exit status, finite outputs, and normalized
primal, stationarity, dual, and complementarity residuals no greater than
$10^{-7}$. The original error-decomposition checks retain limits of
$10^{-11}s_{\mathrm{err}}$ for linear residuals and
$10^{-11}s_{\mathrm{err}}^2$ for squared residuals. Here,
$s_{\mathrm{err}}$ is the greater of one and the largest absolute error
component in the selected record, expressed in the original numerical
coordinates. First-step and complete-affine freezing identities retain the
absolute limit $10^{-11}$. The physical state checks retain the
$10^{-8}\,\mathrm{rad}\,\mathrm{s}^{-1}$ limit and the original integration
and refinement settings.

Directional and interval claims retain their specified sign, ordering,
and zero-containment checks. Two restricted exceptions concern descriptive
algebraic cross terms and specified Study~1 freezing-error contrasts whose
exact value is zero. At one step, the latter exception requires agreement of the nonlinear and
frozen-affine predictions. At longer horizons, the exception requires held
inputs and both compared maps to be affine in the state. Every selected freezing residual must satisfy the existing
$10^{-11}$ limit for both models in both records. The complete quadratic
model is eligible only at one step. These exceptions do not follow from
small magnitude alone and do not assert agreement of statistical
significance.

Recent-regressor Gram diagnostics receive separate matrix, formation,
spectral, rank, and source checks. Here, a passive diagnostic is recorded for analysis but does not enter
the identification or control calculation. The supplementary rank threshold is
$10^{-10}$. A condition estimate is treated as resolved only when its prescribed
numerical screening interval contains the stored and recomputed estimates
and has normalized width no greater than $10^{-6}$. The versioned verification
specification defines the interval and its width normalization. Unresolved values and
reductions can be retained only in the specified passive Gram cases,
with the required paired rank and scientific-outcome checks satisfied.
Such results are \emph{qualified}, not numerically equal. This qualification
does not waive the resolved-condition requirements for the RLS covariance
or QP Hessian. These screening intervals are numerical diagnostics, not
certified spectral bounds or proofs of algebraic rank.

The passive Study~2 QP-condition maximum has a separate locator rule.
Each record must identify its own first maximizing index exactly. Different indices require agreement of the complete histories and maxima.
For each record, the loss is the difference between its own maximum and its
value at the other record's maximizing index. Both losses must lie within
the sum of the corresponding elementwise allowances. No agreement of peak time is asserted.

Derived tables and figures must use the specified source records. Known 15-significant-digit
CSV encoding receives a separate bounded representation allowance after
source validation; rounded manuscript entries alone do not establish
reproduction.
\label{par:cross-platform-verification-end}

\paragraph{Certification records and source provenance}
\label{par:cross-platform-provenance}
\begingroup
\emergencystretch=1em
Both full reproductions used MATLAB 26.1.0.3346908 (R2026a), Update~5,
and Optimization Toolbox~26.1. The macOS system used version~26.6.2
(build~25G83), an Apple M4 processor, Apple Accelerate BLAS (ILP64),
and NAG Performance Components~1.4.1 with LAPACK~3.9.1, with a recorded
maximum of ten computation threads. The Windows system used Windows~11
Home (build~26200), an Intel Core i7-12700K processor, and Intel oneAPI
Math Kernel Library~2025.3 with LAPACK~3.12.1, with a recorded maximum
of twelve computation threads.

Both used \texttt{quadprog} with
\texttt{interior-point-convex} and explicit constraint and optimality
tolerances of $10^{-8}$. Recorded effective options also include
\texttt{MaxIterations} equal to 200 and \texttt{StepTolerance} equal to
$10^{-12}$.
\par\endgroup

The macOS full reproduction used revision \texttt{9855c14210643}, whereas
the Windows full reproduction used revision \texttt{b8754a293810}.\footnote{\raggedright The
full source revisions are
\texttt{9855c14210643cc06908fb529afd47953413ed6c} and
\texttt{b8754a2938104f5843aebd60c4d51162223407f3}, respectively.
The common acceptance-policy SHA-256 is
\texttt{ca9101ea01a6eb09\allowbreak{}3664d383a3d6bd17\allowbreak{}34598ba84aed2ee3\allowbreak{}5b28cfa7412aec58}.}
Only two test files and their verification metadata differ between these
revisions. The simulation and comparison routines, configurations, reference results,
and effective acceptance policy are unchanged between these revisions.
The separately executed amended test suite passed on each platform. The macOS full run was not repeated at the later revision.

Each full reproduction completed 2,073 control runs: 503 primary runs,
12 additional Study~1 diagnostic-support runs, and 1,558 sensitivity runs.
Study~6 recomputed its diagnostics from the corresponding newly generated
source runs. Both reproductions satisfied the prescribed numerical,
source, table, and figure checks, with the stated Gram and
claim-applicability qualifications retained.
Outside the qualified Gram comparisons, the largest reported ratio of a
numerical discrepancy to its allowance was $1.78\times10^{-3}$ on macOS
and zero on Windows. These ratios do not establish numerical agreement
of the qualified Gram quantities. The protected reference records
retained their paths, sizes, and SHA-256 hashes.

The original producing revisions and complete machine and solver records
for Studies~1--6 are not established by the retained evidence. The
certification revisions above identify the later reproductions, not the
original generation of those historical results. The separately recorded
source and solver information for the sensitivity calculation remain as
specified earlier in this appendix.
\label{par:cross-platform-provenance-end}
\par

\Needspace{4\baselineskip}
\paragraph{Code availability}
\label{par:code-availability}
The MATLAB implementation is available at
\url{https://github.com/tamwng/tam-control-toolbox}.
Version~2.0.0 is the designated release for reproduction and source
traceability.

The reported studies provide exploratory numerical evidence under
the stated conditions. No independent confirmation is claimed.